\documentclass[acmsmall,screen]{acmart}
\usepackage{microtype} 
\usepackage{subcaption}

\setcopyright{rightsretained}
\acmPrice{}
\acmDOI{10.1145/3434313}
\acmYear{2021}
\copyrightyear{2021}
\acmSubmissionID{popl21main-p210-p}
\acmJournal{PACMPL}
\acmVolume{5}
\acmNumber{POPL}
\acmArticle{32}
\acmMonth{1}

\newcommand{\RED}[1]{{ \color{red}{#1}}}

\usepackage[T1]{fontenc} 
\usepackage{amsmath}
\usepackage{url}
\usepackage{amsthm}
\usepackage{color}
\usepackage{graphicx}

\usepackage{cmll}
\usepackage{stmaryrd}
\usepackage{proof}
\usepackage{wrapfig}
\newcommand{\bm}[1]{\mathbf{#1}}

\usepackage[english]{babel}
\usepackage{xspace}
\usepackage{cancel}

\newcommand{\deff}{\, := \,}

\newcommand{\mex}{D}
\newcommand{\letr}{\mathtt{let}}
\newcommand{\lett}[2]{\mathtt{let~} x=#1 \mathtt{~in~} #2}
\newcommand{\plett}[3]{\mathtt{let~} #1=#2 \mathtt{~in~} #3}

\newcommand{\ee}{{\bf E}}

\renewcommand{\ss}{\bm S}

\newcommand{\seq}[1]{\{#1_n\}_{n\in\Nat}}

\newcommand{\cbv}{{\mathtt{cbv}}}
\newcommand{\cbn}{{\mathtt{cbn}}}
\newcommand{\cps}{{\mathtt{cps}}}

\newcommand{\red}{\rightarrow}

\newcommand{\full}{\rightrightarrows}

\newcommand{\Red}{\Rightarrow}

\renewcommand{\int}{\textsf {int}}

\newcommand{\PLambda}{\Lambda_\oplus}

\newcommand{\Val}{\mathcal V}

\newcommand{\st}{\mid}

\newcommand{\two}{\frac{1}{2}}
\newcommand{\four}{\frac{1}{4}}

\newcommand{\Nat}{\mathbb{N}}
\newcommand{\Qnum}{\mathbb{Q}}

\newcommand{\TVar}{\textsc{Var}}
\newcommand{\TVal}{\textsc{Val}}
\newcommand{\TZero}{\textsc{Zero}}

\newtheorem{theorem}{Theorem}
\newtheorem{lemma}[theorem]{Lemma}
\newtheorem{property}[theorem]{Property}

\newtheorem{fact}[theorem]{Fact}

\newtheorem{example}[theorem]{Example}
\newtheorem{remark}[theorem]{Remark}

\newtheorem{thm}[theorem]{Theorem}
\newtheorem{prop}[theorem]{Proposition}
\newtheorem{cor}[theorem]{Corollary}

\theoremstyle{definition}
\newtheorem{Def}[theorem]{Definition}

\newtheorem*{theorem*}{Theorem}
\newtheorem*{proposition*}{Proposition}
\newtheorem*{prop*}{Proposition}
\newtheorem*{lemma*}{Lemma}
\newtheorem*{remark*}{Remark}
\newtheorem*{notation*}{Notation}
\newtheorem*{fact*}{Fact}

\newcommand{\lam}{\lambda}

\newcommand{\ie}{\emph{i.e.}\xspace}
\newcommand{\eg}{\emph{e.g.}\xspace}
\newcommand{\ih}{\emph{i.h.}\xspace}

\usepackage{xcolor}

\newcommand{\blue}[1]{{\color{blue}{#1}}}
\newcommand{\green}[1]{{\color{green}{#1}}}

\newcommand{\DST}[1]{\mathtt{DST}(#1)}

\newcommand{\MDST}[1]{\mathtt{MDST}(#1)}

\newcommand{\m}{\mathtt m}

\newcommand{\norm}[1]{\|#1\|}
\newcommand{\nnorm}[1]{\|#1^{\Val}\|}

\newcommand{\iI}{i \in I}
\newcommand{\jJ}{j \in J}
\newcommand{\kK}{k  \in K}

\makeatletter
\newcommand{\oset}[3][0ex]{%
	\mathrel{\mathop{#3}\limits^{
			\vbox to#1{\kern0\ex@
				\hbox{$\tiny#2$}\vss}}}}
\makeatother
\usepackage{floatflt}
\usepackage[font=scriptsize,skip=4pt]{caption}

\newcommand{\ma}{\mathcal A}
\newcommand{\mb}{\mathcal B}

\newcommand{\mA}{\mathcal A}
\newcommand{\mB}{\mathcal B}
\newcommand{\mC}{\mathcal C}
\newcommand{\at}{\mathtt a}
\newcommand{\bt}{\mathtt b}
\newcommand{\ct}{\mathtt c}

\newcommand{\At}{\mathtt A}
\newcommand{\Bt}{\mathtt B}

\newcommand{\arrow}{\rightarrow}
\newcommand{\der}{\vdash}
\newcommand{\dem}{~\triangleright~}

\newcommand{\size}[1]{|#1|}
\newcommand{\sz}[1]{\size{#1}}

\usepackage{proof}

\newcommand{\dsum}{\bigsqcupplus}
\newcommand{\dist}[1]{\left\langle #1 \right\rangle}

\newcommand{\mdist}[1]{\left\langle #1 \right\rangle}

\newcommand{\mset}[1]{\left[#1\right]}
\newcommand{\mul}[1]{[#1]}
\newcommand{\zero}{\mathbf 0}

\newenvironment{varitemize}
{
	\begin{list}{\labelitemi}
		{\setlength{\itemsep}{0pt}
			\setlength{\topsep}{0pt}
			\setlength{\parsep}{0pt}
			\setlength{\partopsep}{0pt}
			\setlength{\leftmargin}{15pt}
			\setlength{\rightmargin}{0pt}
			\setlength{\itemindent}{0pt}
			\setlength{\labelsep}{5pt}
			\setlength{\labelwidth}{10pt}
	}}
	{
	\end{list}
}
\newcounter{numberone}

\newcommand{\ovdash}[1]{\overset{\blue{#1}}{\vdash}}
\newcommand{\oder}{\ovdash}

\newcommand{\subs}[2]{ \{#2{/}#1\} }

\renewcommand{\for}{\pmb{.}}

\newcommand{\PAST}{\texttt {PAST}\xspace}
\newcommand{\AST}{\texttt {AST}\xspace}

\newcommand{\ProbTerm}[1]{\mathtt{PTerm}\big(#1\big)}
\newcommand{\ETime}[1]{\mathtt{ETime}\big(#1\big)}

\newcommand{\eval}[2]{{\mathtt{PTerm}_{#1} \big(#2\big)}}
\newcommand{\etime}[2]{\mathtt{ETime}_{#1}\big(#2\big)}

\newcommand{\X}{\mathbb X}

\newcommand{\pPi}[2]{\Pi^{#1}_{#2}}
\newcommand{\pSigma}[2]{\Sigma^{#1}_{#2}}

\newcommand{\type}{\mathbf T}
\renewcommand{\DST}[1]{\mathcal{D}(#1)}
\renewcommand{\MDST}[1]{\mathcal{M}(#1)}
\newcommand{\MPLambda}{\MDST{\PLambda}}

\renewcommand{\Red}{\full}
\newcommand{\PLambdaCBV}{\PLambda^\cbv}

\newcommand{\version}{1}
\newcommand{\commented}{0} 

\newcommand{\condinc}[2]{\ifthenelse{\equal{\commented}{0}}{#1}{\green {**[[#2]]**}}}

\newcommand{\SLV}[2]{\ifthenelse{\equal{\version}{0}}{#1}{ #2}}

\begin{document}

\author{Ugo Dal Lago}
\affiliation{\institution{Universit\`a di Bologna}\country{Italy}}
\additionalaffiliation{\institution{INRIA Sophia Antipolis}}
\email{ugo.dallago@unibo.it}

\author{Claudia Faggian}
\affiliation{\institution{Universit\'e de Paris, IRIF, CNRS, Paris, France}\country{France}}
\email{claudia.faggian@irif.fr}

\author{Simona Ronchi Della Rocca}
\affiliation{\institution{Universit\`a di Torino}\country{Italy}}
\email{ronchi@di.unito.it}

\title{Intersection Types and (Positive) Almost-Sure Termination}

\begin{abstract}
  Randomised higher-order computation can be seen as being captured by
  a $\lambda$-calculus endowed with a single algebraic operation,
  namely a construct for binary probabilistic choice. What matters
  about such computations is the \emph{probability} of obtaining any
  given result, rather than the \emph{possibility} or the
  \emph{necessity} of obtaining it, like in (non)deterministic
  computation. Termination, arguably the simplest kind of reachability
  problem,  can be spelled out in at least two ways, depending on
  whether it talks about the probability of convergence or about the
  expected evaluation time, the second one providing a stronger
  guarantee. In this paper, we show that intersection types are
  capable of precisely characterizing \emph{both} notions of
  termination inside \emph{a single} system of types: the probability
  of convergence of any $\lambda$-term can be underapproximated by its
  \emph{type}, while the underlying derivation's \emph{weight} gives a
  lower bound to the term's expected number of steps to normal
  form. Noticeably, both approximations are tight---not only soundness
  but also completeness holds.  The crucial ingredient is
  non-idempotency, without which it would be impossible to reason on
  the expected number of reduction steps which are necessary to
  completely evaluate any term. Besides, the kind of approximation we
  obtain is proved to be \emph{optimal} recursion theoretically: no
  recursively enumerable formal system can do better than that.
\end{abstract}



\maketitle

\section{Introduction}
The study and analysis of randomised computation is almost as old as
theoretical computer science
itself~\cite{de1956computability,Rabin63,Santos69}. In randomised computation,
algorithms may well violate determinism by performing some inherently
stochastic operations, like the one consisting in triggering
probabilistic choice. In the last fifty years, randomised computation
has been shown to enable efficient algorithms~\cite{Motwani}, but also
secure cryptographic primitives (e.g. public-key
cryptosystems~\cite{GoldwasserMicali}), which are provably impossible
to define in a purely deterministic computational model.

Research on programming languages featuring various forms of random
choice operators has itself a long
history~\cite{kozen1981,SahebDjahromi1978}, but has shown a strong
impetus in the last ten years, due to progress in so-called bayesian
programming languages~\cite{gmrbt2008,TMW2015}, in which not only
probabilistic \emph{choice} is available, but also \emph{conditioning}
has a counterpart inside programs, usually in the form of
\texttt{observe} or \texttt{score} statements. In an higher-order
scenario, the mere presence of a probabilistic choice operator,
however, poses a number of challenges to the underlying theory.
For example, relational reasoning by way of systems of logical
relations~\cite{BB2015}, or by way of coinduction~\cite{DLSA14} has
proved to be possible, although requiring some new ideas, both
definitionally, or in the underlying correctness proof. Moreover,
giving a satisfactory denotational semantics to higher-order languages
with binary probabilistic choice is notoriously
hard~\cite{JonesPlotkin1989,jungtix1998}, and has been solved in a
completely satisfactory way only relatively
recently~\cite{ETP14,goubault2015}.

\paragraph{Types and Verification.}
Verification of deterministic higher-order programs can be carried out
in many ways, including model checking~\cite{Ong2006}, abstract
interpretation~\cite{Cousot1997}, and type
systems~\cite{Pierce2002}. Among the properties one is interested in
verifying programs against, safety and reachability are arguably the
simplest ones. Type systems, traditionally conceived as lightweight
methodologies ensuring safety (hence the slogan ``well-typed programs
cannot go wrong''), can also be employed to check reachability and
termination~\cite{HPS1996,UrzyczynLecturesCurryHoward}. This
idea has been brought to its extreme consequences by the line of work
on intersection types~\cite{CoppoDezani1978,CDV81}, which not only \emph{guarantee} termination,
but also \emph{characterise} it, this way providing a compositional
presentation of \emph{all} and \emph{only} the terminating
programs. Indeed, intersection types can be seen as giving
\emph{semantics} to higher-order programs~\cite{BCDC1983}, and
also to support program verification in subrecursive
languages~\cite{Kobayashi2009}.

\paragraph{On Probabilistic Termination's Double Nature.}
But what it means for a \emph{probabilistic} program to
\emph{terminate} or---slightly more generally---to \emph{reach a
  state} in which certain conditions hold? A first answer consists in
considering a program terminating if the probability of divergence is
null, namely if the program is \emph{almost-surely terminating} (\AST
for short). This way, even when the possibility of diverging is still
there, it has null probability. This, however, does not mean that
\emph{the time} to termination (better, the \emph{expected} time to
termination) is finite: this is a stronger and computationally more
meaningful requirement, called \emph{positive}\footnote{The term was
  introduced in \cite{BournezG05}, but the requirement that the program be
  expected to terminate is natural and fundamental, and was already
  present in \cite{SahebDjahromi1978}.} almost-sure termination (shortened to \PAST
in the following).  It is in fact well-known that checking programs
for (positive) almost-sure termination turns out to be strictly
harder, recursion theoretically, than checking termination of
deterministic programs~\cite{KKM2019}: both almost-sure termination and
positive almost-sure termination are not recursively enumerable, and
have \emph{incomparable} recursion-theoretic statuses, the former
being $\pPi{0}{2}$-complete, the latter being $\pSigma{0}{2}$-complete. The
discrepancy with the realm of deterministic calculi
can be seen also in sub-universal languages: recently, Kobayashi, Dal
Lago and Grellois~\cite{KDLG2019}, have shown that model checking
reachability properties is undecidable in probabilistic higher-order recursion
schemes, while the same problem is
well known to be decidable in their deterministic and nondeterministic
siblings~\cite{Ong2006}. More generally, the nature of
probabilistic termination in presence of higher types
is still not completely understood, and is fundamentally different
from the one of its deterministic counterpart.

\paragraph{Some Natural Questions.}
Given the rich theory that  the programming language community has been able
to build for the deterministic $\lambda$-calculus, a number of
questions naturally arise. Is it possible to \emph{faithfully}
and \emph{precisely} reflect the expected time to termination by a system of types?
What are the limits to the expressive power of such a system, given
the aforementioned recursion theoretic limitations? Do intersection
types can be of help, given their successes in characterising various
notions of termination in a deterministic setting?  These questions
are natural ones, but have remained unanswered so far. This paper is
the first one giving answers to them.

\paragraph{Contributions.}
We show here that intersection types indeed capture \emph{both forms}
of probabilistic termination in untyped probabilistic
$\lambda$-calculi. More specifically, we define a system of
non-idempotent intersection types such that from any type derivation
for a given term $M$, one can extract (in an effective, and even
efficient, way) both a lower bound to the \emph{expected time} to termination
for $M$, and a lower bound to $M$'s \emph{probability} of termination.
Remarkably, both kinds of bounds are tight, i.e.  for every
$\varepsilon>0$ there is a type derivation for $M$ which gives an
$\varepsilon$-precise bound to both the probability of and the
expected time to termination.  The main novelty of the proposed
methodology is the presence of distinct ingredients \emph{within} the
same type system, namely monadic types~\cite{DLG2019}, intersection
types~\cite{CoppoDezani1978}, and
non-idempotency~\cite{Carvalho2018}. Their contemporary presence  forces us to switch from a purely qualitative notion of
intersection (i.e.  multisets) to a quantitative one (i.e. scaled
multisets). This is necessary to appropriately deal with the multiple
uses of program variables in presence of probabilistic choice. In
view of the non-recursive enumerability of either kinds of
probabilistic termination, taking type derivations as
\emph{approximate} witnesses to termination, rather than \emph{proper}
ones, indeed makes sense, and is the best one can do: we prove
that any (recursively enumerable) system of types for a probabilistic
$\lambda$-calculus is either unsound or incomplete as a way to
\emph{precisely} verify termination properties of pure
$\lambda$-terms. In other words, one cannot do better
than what we do. Remarkably, all results we give in this paper 
hold for both call-by-value and  call-by-name evaluation, but we prefer to
give all the details of the a system of the former kind, arguably a
more natural one in presence of effects.

\SLV{An extended version of this paper with proofs and more
	details is available~\cite{LVarxiv}.}
{This report is an extended version of \cite{DLFRpopl21}.}


\section{A Gentle Introduction to Intersection Types, Termination, and Randomization}\label{sec:gentleintro}

This section is meant to introduce the non-specialist to intersection
types\footnote{Our introduction to  non-idempotent intersection types  is inspired by that in \cite{KesnerV20}} seen as a characterisation of terminating deterministic
programs, and to the challenges one faces when trying to generalise
intersection types to calculi featuring binary probabilistic
choice.
\subsection{Intersection Types and Termination}\label{sec:intro_IT}
Suppose we work within a simple functional programming language,
expressed as a call-by-value (CbV) $\lambda$-calculus $\Lambda^\cbv$
in A-normal form.  Values and terms are generated through the
following grammars:
\begin{align*}
  V &::= x \mid \lambda x.M & \mbox{\textbf{Values}, }\Val^\cbv\\
  M &::= V \mid VV \mid \lett M M & \mbox{\textbf{Terms}, }\Lambda^\cbv
\end{align*}
Evaluation of closed terms is captured by two reduction rules, namely
$(\lambda x.M)V\rightarrow M\{V/x\}$ and $\lett V M\rightarrow
M\{V/x\}$, which can be applied in any evaluation contexts, i.e. in
any expression from the grammar $E::=[\cdot]\mid\lett E M$.  As customary when working with
functional languages, evaluation is weak (\ie, no reduction can take
place in the body of a $\lambda$-abstraction).

This language can be seen as a fragment of Plotkin's CbV
$\lambda$-calculus~\cite{PlotkinCbV} in which the latter can be faithfully
embedded\footnote{an application $MN$ becomes the term
  $\plett{x}{M}{\plett{y}{N}{xy}}$.}. As such, the calculus is easily
seen to be Turing-universal, and  termination is thus an
undecidable---although recursively enumerable---problem.  How could we
\emph{compositionally} guarantee termination of those $\lambda$-terms?
The classic answer to the question above consists in endowing the
calculus with a system of types. As an example, a system of \emph{simple}
types for the terms in $\Lambda^\cbv$ is in Figure~\ref{fig:simpletypes}, where
types are either an atom $\alpha$ or an arrow type
$\At\rightarrow\Bt$.  A simple reducibility-like argument indeed shows
that typability ensures termination.
\begin{figure*}\centering
    \begin{minipage}{.97\textwidth}
      \vspace{10pt}
      \begin{minipage}{.35\textwidth}
        $$
        \At::=\alpha\mid\At\rightarrow\At
        $$
      \end{minipage}
    \begin{minipage}{.6\textwidth}
   {\small  $$
    \infer{\Gamma,x:\At\vdash x:\At}{}
    \qquad
    \infer{\Gamma\vdash\lambda x.M:\At\rightarrow \Bt}{\Gamma,x:\At\vdash M:\Bt}
    $$
    \vspace{-5pt}
    $$
    \infer{\Gamma\vdash VW:\Bt}{\Gamma\vdash V:\At\rightarrow\Bt & \Gamma\vdash W:\At}
    \qquad
    \infer{\Gamma\vdash \lett{N}{M}:\Bt}{\Gamma\vdash N:\At & \Gamma,x:\At\vdash M:\Bt}    
    $$}
    \end{minipage}
  \end{minipage}
  \caption{Simple Types}\label{fig:simpletypes}
\end{figure*}
The converse does not hold, i.e. simple types are highly incomplete as
a way to type terminating terms. As an example, self application,
namely the value $\lambda x.xx$, is not simply-typable even if terminating,
since the variable $x$ cannot be assigned both the type $\At$ and the type
$\At\rightarrow\Bt$.

One way to go towards a type system complete for termination consists
in resorting to some form of polymorphism. For example, \emph{parametric}
polymorphism in the style of System $\mathbb{F}$~\cite{Girard1971} dramatically
increases the expressive power of simple types by way of a form of
(second-order) quantification: the type $\forall\alpha.\At$
stands for all types which can be obtained as formal instances of
$\At$. Parametric polymorphism, however, is not enough to get to a
complete system, which can instead be built around \emph{ad-hoc} polymorphism:
rather than extending simple types by way of quantifiers,
one can enrich types with intersections in the form of
\emph{finite sets} of types $\mA=\{\At_1,\ldots,\At_n\}$, and take
arrow types as expressions in the form $\mA\rightarrow\mB$. The type
$\mA$ can be assigned to terms which have type $\At_j$ \emph{for
  every} $j\in\{1,\ldots,n\}$. The resulting type
system is in Figure~\ref{fig:intersectiontypes}, and is well-known to
be both sound \emph{and complete} for termination. 
\begin{figure*}\centering
    \begin{minipage}{.97\textwidth}
      \vspace{7pt}
      \begin{minipage}{.35\textwidth}
        \begin{align*}
          \At&::=\mA\rightarrow\mA\\
          \mA&::=\{\At_1,\ldots,\At_n\}
        \end{align*}
        \end{minipage}
    \begin{minipage}{.6\textwidth}
   {\small  $$
    \infer{\Gamma,x:\mA\vdash x:\mA}{}
    \qquad
    \infer{\Gamma\vdash\lambda x.M:\mA\rightarrow \mB}{\Gamma,x:\mA\vdash M:\mB}
    \qquad
    \infer{\bigcup_i\Gamma_i\vdash V:\{\At_i\}_{i\in I}}{\{\Gamma_i\vdash V:\At_i\}_{i\in I}}
    $$
    \vspace{-5pt}
    $$
    \infer{\Gamma\vdash VW:\mB}{\Gamma\vdash V:\mA\rightarrow\mB & \Gamma\vdash W:\mA}
    \qquad
    \infer{\Gamma\vdash \lett{N}{M}:\mB}{\Gamma\vdash N:\mA & \Gamma,x:\mA\vdash M:\mB}    
    $$}
    \end{minipage}
  \end{minipage}
  \caption{Idempotent Intersection Types for $\Lambda^\cbv$}\label{fig:intersectiontypes}
\end{figure*}

There is even more.  One can make type derivations capable of
reflecting \emph{quantitative} kinds of information such as the number of
required evaluation steps, rather than merely termination (which
is \emph{qualitative} in nature). This requires taking intersection
types not as sets, but rather as \emph{multi}sets, i.e.
$\mA=[\At_1,\ldots,\At_n]$. This form of intersection type is dubbed
\emph{non-idempotent}, due to the non-idempotency of multiset
unions and intersections.  Type environments need now be treated multiplicatively
rather than additively, this way giving a linear flavour to the type
system. In non-idempotent intersection types,  a natural
number $w$ can be assigned to any type derivation 
in such a way that $\ovdash{w}M:[]$  (where $[]$ is the empty multiset seen as an intersection type)
\emph{if and only if} $M$ can be reduced to normal form in \emph{exactly} $w$ steps.
The resulting system is in Figure~\ref{fig:nonidempintersectiontypes}, and
is essentially the one from~\cite{AGL19}.
\begin{figure*}\centering
    \begin{minipage}{.97\textwidth}
      \begin{minipage}{.25\textwidth}
        \begin{align*}
          \At&::=\mA\rightarrow\mA\\
          \mA&::=[\At_1,\ldots,\At_n]
        \end{align*}
        \end{minipage}
    \begin{minipage}{.7\textwidth}
  {\small   $$
    \infer{x:\mA\ovdash{0} x:\mA}{}
    \qquad
    \infer{\Gamma\ovdash{w+1}\lambda x.M:\mA\rightarrow \mB}{\Gamma,x:\mA\ovdash{w} M:\mB}
    \qquad
    \infer{\uplus_i\Gamma_i\ovdash{\sum_{i}w_i} V: [\At_i]_{i\in I}}{\{\Gamma_i\ovdash{w_i} V:\At_i\}_{i\in I}}
    $$
    $$
    \infer{\Gamma\uplus\Delta\ovdash{w+v} VW:\mB}{\Gamma\ovdash{w} V:\mA\rightarrow\mB & \Delta\ovdash{v} W:\mA}
    \qquad
    \infer{\Gamma\uplus\Delta\ovdash{w+v+1} \lett{N}{M}:\mB}{\Gamma\ovdash{w} N:\mA & \Delta,x:\mA\ovdash{v} M:\mB}    
    $$}
    \end{minipage}
  \end{minipage}
  \caption{Non-Idempotent Intersection Types for $\Lambda^\cbv$}\label{fig:nonidempintersectiontypes}
\end{figure*}

\subsection{Typing Termination in a Probabilistic Setting}
How about probabilistic $\lambda$-calculi? Can the story in
Section~\ref{sec:intro_IT} be somehow generalised to such calculi?
Endowing the class of terms with an operator for
fair\footnote{Accommodating an operator for general binary
  probabilistic choice (e.g. in the form $\oplus_q$, where $q$ is a
  rational between $0$ and $1$) would be harmless, but would result in
  heavier notation; we thus prefer to stick to the fair case.} binary
probabilistic choice is relatively easy: the grammar of terms needs to
be extended by way of the production $M::= M\oplus M$, and the term
$M\oplus N$ evolves to either $M$ or $N$ with probability
$\frac{1}{2}$, turning reduction on terms from a deterministic
transition system to a Markov Chain with countably many states. Let us
illustrate all this by way of an example, which will be our running
example throughout the paper.
\begin{example}[Running Example]\label{ex:running}\label{ex:main}
  Let us consider the term $\mex\mex$ where $\mex=\lam x. (xx\oplus
  I)$, and $I$ is the identity $\lambda y.y$. The program $\mex\mex$
  reduces to $\mex\mex\oplus I $, which in turn reduces to either
  $\mex\mex$ or to $I$ with equal probability $1/2$. It is easy to see
  that after $2n$ steps, $\mex\mex$ has terminated with probability
  $\sum_1^n \frac{1}{2^n}$: while running $\mex\mex$, only
  one among the $2^n$ possible outcomes of the $n$ 
  coin-flips results in the term staying at $\mex\mex$, all the others
  leading to $I$.  Noticeably, the expression above tends to $1$ when
  $n$ tends to infinity.  By weighting the steps with their
  probability, we have that the expected number of steps for
  $\mex\mex$ to terminate is $4$. In other words $\mex\mex$ is not
  only almost-surely terminating, but positively so.
\end{example}
As this example shows, despite the minimal changes to the underlying
operational semantics, \emph{reasoning} about randomised computations
can be more intricate than in the usual deterministic setting. More specifically:
\begin{varitemize}
\item
  \emph{Output}. While a deterministic program maps inputs to outputs, a
  probabilistic program maps inputs to \emph{distributions} of
  outputs. For example, $\mex\mex$ evaluates to 
  the Dirac distribution where all the probability is concentrated in the term $I$. 
 Notice that this level of certitude is reached only
  at the limit, not  in any finite amount of steps.
\item
  \emph{Termination}. A deterministic program either terminates on a given
  input or not.  As we mentioned in the Introduction, a probabilistic
  program may give rise to diverging runs, still being almost-surely terminating.
  This is precisely what happens when evaluating $\mex\mex$:
  there is one run, namely the one always staying at $\mex\mex$, which
  diverges, but this run has of course null probability.
\item
  \emph{Runtime}. If a deterministic program terminates, it reaches
  its final state in finitely many steps, and we interpret this number as
  the \emph{time to} termination. In the probabilistic case, what interests
  us is rather the \emph{expected} number of steps, that is the average number
  of steps of the program's runs. Such expected value
  may or may not be finite, even in the case of \AST programs.
  When evaluating $\mex\mex$, this number is finite, but it arises (once again)
as   the sum of an infinite numerical series.
\end{varitemize}

\newcommand{\mexvar}{C} 
\newcommand{\cex}{C}
\newcommand{\SUCC}{\mathit{SUCC}}
\newcommand{\const}{s}

\renewcommand{\succ}[1]{{\mathbf{succ}}(#1)}
\newcommand{\succn}[2]{{\mathbf{succ}}^{#1}(#2)}
\newcommand{\nzero}{\overline{0}}

\newcommand{\eex}{E}
\newcommand{\EXPL}{\mathit{EXP}}
\newcommand{\expl}[1]{{\mathbf{exp}}(#1)}

Small variations on Example~\ref{ex:main} are sufficient to obtain
terms whose behavior is more complex than that of $\mex\mex$.  The
following example illustrate that a term $M$ can reach
\emph{countably} many distinct normal forms and intermediate values, and
that almost-sure termination does \emph{not} imply \emph{positive} almost-sure termination.
\begin{example}\label{ex:succ}
  For every natural number $n$, let $\overline{n}$ be an encoding of
  it as a $\lambda$-term, and let $\SUCC$ and $\EXPL$ be terms which
  encode the successor and the exponential function, respectively.
  \begin{varitemize}
  \item
    Consider the term $\cex\cex$ where $\cex=\lam x.(\succ{xx}\oplus
    \nzero)$, and where  $\succ M$ is  syntactic sugar for $(\plett z M
          {\SUCC\,z})$. Note that $\succ{\overline{n}}$
          reduces to $\overline{n+1}$ in constant time $\const$ ($\const\in \Nat$), of course
          depending on the chosen encoding.  The program $\cex\cex$
          reduces---at the limit---to each natural number
          $\overline{n}$ with probability $\frac{1}{2^{n+1}}$. It is
          clear that $\cex\cex$ is \AST, and it is easy to check that
          it is also \PAST; indeed it is expected to terminate in
          $4+ \const$ steps.  However its reduction graph, contrarily to the
          one of $\mex\mex$, involves \emph{infinitely many normal
            forms}.
  \item
    The term $\expl{CC}$, where $\expl M$ is syntactic sugar for $(\plett z M
    {\EXPL\, z})$, is a term which is still almost-surely
      terminating, but \emph{not positively}. Indeed, its expected runtime
    is infinite.
\end{varitemize}
\end{example}
\newcommand{\mexf}{\mex_1}
\newcommand{\mexa}{\mex_2}

All this shows that typing probabilistically terminating
programs requires us to go  significantly beyond classic
intersection type theories, but also beyond the few attempts on
	type theories for probabilistic $\lambda$-calculi in the literature.

Let us now take a look at how the term $\mex\mex$ could be given an
intersection type, in a way reflecting its being (positively)
almost-surely terminating. Let us write $\mex\mex$ as $\mexf
\mexa$. The term $\mexf$ uses its argument $\mexa$ in two different
ways, the first as a function and the second as an argument to the
same function.  We already know that intersection types are there
precisely for this purpose.  But there are some fundamental
differences here compared to the deterministic case: first of all, the
two copies of $\mexa$ that the function $\mexf$ consumes are
\emph{used} only with probability $\frac{1}{2}$. Moreover, $\mexf$
returns \emph{two} different results, namely $I$ and $\mexa\mexa$,
each with equal probability. These two observations inform how
non-idempotent intersection types can be generalised to a
$\lambda$-calculus with probabilistic choice. Indeed, the multisets
$\mA$ and $\mB$ in an arrow type $\mA\rightarrow\mB$ have to be
\emph{enriched} with some quantitative information:
\begin{varitemize}
\item
  in order to capture the \emph{termination probability}, the
    intersection type $\mB$ needs to be turned into a
    \emph{distribution} of intersection types, reflecting the fact
    that the output of a computation is not \emph{one} single value,
    but rather \emph{a distribution} of them.
\item
  capturing \emph{time expectations} requires typing to become even
  more sophisticated, introducing two novelties:
  \begin{varitemize}
  \item
    The multiset of types $\mA$ needs to carry some information about
    the probability of each copy of the argument to be
    actually used. In other words, elements of $\mA$ needs to be
    \emph{scaled}. Note the discrepancy between the ways $\mA$ and
    $\mB$ are treated: in the former a form of scaled multiset
    suffices, while in the latter a distribution of intersection types
    is needed.
  \item
    Moreover, the type system needs to be capable of dubbing
    divergent terms as having \emph{arbitrarily large}
    evaluation time expectations. Consider, as an example,
    the program $M=I\oplus \Delta\Delta$, where
    $\Delta=\lam x.xx$. In one evaluation step, such a term reduces to
    the value $I$ with probability $\two$ or to the diverging term
    $\Delta\Delta$, with equal probability $\two$.  The expected
    runtime of $M$ is therefore infinite: $1+\sum_{i=1}^\infty\two$.
    Since typing $M$ requires giving a type to $\Delta\Delta$, the
    latter has to be attributed arbitrary large weights, although
    the only type it can receive is for obvious reasons the empty
    distribution.
  \end{varitemize}
  We come back to all this in Section~\ref{sec:examples_CbV}, after formally
  introducing the type system.
\end{varitemize}
The aforementioned ones are not the only novelties of the type system
we introduce in this paper.  Given the already mentioned results by Kaminski et
al. on the hardness of probabilistic
termination~\cite{KKM2019}, in which both notions of
termination are proved \emph{not} to be recursively enumerable, there
is simply no hope to obtain results like the classic ones on deterministic
terms, in which correctness of \emph{one} derivation 
 can serve as a termination certificate (this, to be
fair, if checking type derivations for correctness remains decidable).
The way out consists in looking at a characterisation by way of
\emph{approximations}: a type derivation would not be a witness of
(positive) almost-sure termination by itself, but a witness of some
\emph{lower-bound} on the probability of termination or on the
expected number of steps to termination. The type system needs to be
tailored for this purpose.


\section{A Probabilistic Call-by-Value $\lambda$-Calculus} \label{sec:cbv}
In this section, we formally introduce the minimalistic probabilistic
functional programming language we have sketched in
Section~\ref{sec:intro_IT}, and that we indicate in the following as $\PLambdaCBV$.
We start with some technical definitions, which we will use throughout
the paper.
\subsection{Mathematical Preliminaries}\label{sec:preliminariesM}
\paragraph{Multisets.}
We denote a \emph{finite multiset} (over a set $\X$) as $\mset{\at _j }_{\jJ}$, where the index set $J$ is finite and possibly empty. The
 empty multiset is denoted as $\mset{}$, while elements of a non-empty
 multiset are often enumerated, like in
 $\mset{a,b,c}$. Multiset union is noted $\uplus$.

 \paragraph{Distributions.}
Let $\Omega$ be a \emph{countable} set.  A function
$\mu:\Omega\to[0,1]$ is a probability \emph{subdistribution} if its
\emph{norm} $\norm \mu := \sum_{\omega\in \Omega} \mu(\omega)$ is less
or equal to $1$.  It is a \emph{distribution} if $\norm \mu= 1$.
Subdistributions are the standard way to deal with possibly diverging
probabilistic computations.  We write $\DST{\Omega}$ for the set of
subdistributions on $\Omega$, equipped with the standard pointwise
partial order relation : $\mu \leq \rho$ if $\mu (\omega) \leq \rho
(\omega)$ for each $\omega\in \Omega$.  The \emph{support} of $\mu$ is
the set $\{\omega\mid \mu(\omega)>0\}$.

\paragraph{Multidistributions.}
Suppose $\X$ is a countable set and let $\m$ be a finite multiset of
pairs of the form $pM$, with $p\in(0,1]$, and $M\in \X$. Then
  $\m=\mset{p_iM_i}_{\iI}$ is said to be a \emph{multidistribution
      on $\X$} if $ \norm \m := \sum_{\iI} p_i \leq 1$.  For
    multidistributions, we use the notation $
    \m=\mdist{p_iM_i}_{\iI}$. The empty multidistribution is indicated as
    $\zero$ (note that $\norm \zero =0 $).  We denote by $\MDST \X$
    the set of all multidistributions on $\X$.  We indicate the
    multidistribution $\mdist{1M}$ simply as $\mdist{M}$.  The
    (disjoint) sum of multidistributions is denoted as $\dsum$, and is a \emph{partial} operation.  The
    product $q\cdot \m$ of a scalar $q<1$ and a multidistribution $\m$
    is defined pointwise: $q\cdot \mdist{ p_{1}M_{1},\ldots,
      p_{n}M_{n}}=\mdist{ (qp_{1})M_{1},\ldots, (qp_{n})M_{n}} $.
    Intuitively, a multidistribution $\m$ is an intensional
    representation of a probability distribution: multidistributions
    do not satisfy the equation $\m=p\cdot\m\dsum(1-p)\cdot\m$. This
    being said, every multidistribution can be made to collapse to
    a distribution, by taking the sum of all of its elements referring
    to the same $M\in\X$.

\subsection{The Language $\PLambdaCBV$}
This section is devoted to introducing the language.
\emph{Values} and \emph{terms} are defined  by the grammar
\begin{align*}
V &::= x \mid \lambda x.M &\mbox{\textbf{Values}, }\Val_\oplus^\cbv\\
M &::= V \mid VV \mid  M \oplus M \mid \lett M M &\mbox{\textbf{Terms}, }\PLambdaCBV
\end{align*}
where $x$ ranges over a countable set of \emph{variables}.  $\PLambdaCBV$
and $\Val_\oplus^\cbv$ denote respectively the set of terms and of values.  Free
and bound variables are defined as usual, while $M \subs x N$ denotes the
term obtained from the capture-avoiding substitution of $N$ for all the free
occurrences of $x$ in $M$.  As usual, a \emph{program} is a closed
term. Throughout the paper we frequently use the following terms as examples: 
\[I:= \lam x.x;\qquad\qquad
\Delta:= \lam x. xx;\qquad\qquad
\mex:=\lam x. (xx\oplus I).
\]
The program $\Delta\Delta$ is the paradigmatic diverging term, while
$\mex\mex$ is our running example.
\subsection{The Operational Semantics}
The operational semantics of $\PLambdaCBV$ is formalized through the
notion of multidistribution as introduced in
Section~\ref{sec:preliminariesM}, following~\cite{ADLY2020}. To understand
why this is a convenient way to describe the probabilistic dynamics
of programs, let us consider how terms in $\PLambdaCBV$ could be evaluated.

The intended dynamics of the term $M\oplus N$ is that it reduces to either $M$
or $N$, \emph{with equal probability} $\two$. That is, the state of
the program after one reduction step is $M$ with probability $\two$
and $N$ with probability $\two$. Consider, as an example, the term
$(I\oplus(II))\oplus II$. Its evaluation is graphically represented in
Figure~\ref{fig:IOIIOI}.
\begin{figure}[ht]
  \centering
    \begin{minipage}{.97\textwidth}
      \centering
      \begin{subfigure}{.47\textwidth}
        \centering\includegraphics[scale=0.81]{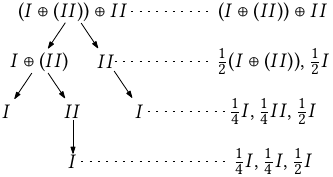}
        \caption{Evaluating $(I\oplus II)\oplus II$}\label{fig:IOIIOI}
      \end{subfigure}
      \begin{subfigure}{.47\textwidth}
        \centering\includegraphics[scale=0.81]{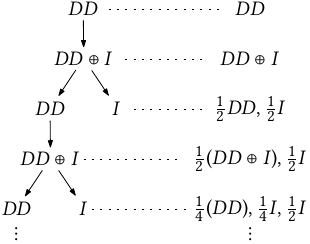}
        \caption{Evaluating $\mex\mex$}\label{fig:DD}
      \end{subfigure}
  \end{minipage}
  \caption{Evaluating some Terms in $\PLambdaCBV$}
\end{figure}
The first computation step consists in performing a probabilistic
choice, proceeding as $I\oplus(II)$ or as $II$ according to its
outcome. While the latter branch ends up in $I$ (which is a value)
in one deterministic step, the former branch proceeds with
another probabilistic choice, which results in either $I$ or $II$.
Finally, after another reduction step, $II$ is reduced to the identity.
To the right of the reduction tree in Figure~\ref{fig:IOIIOI}, one can
see, for each time step, a summary of the ``status'' of all
probabilistic branches, each paired with its probability.  After three steps,
all branches reduce to $I$, and indeed the probability of observing
$I$ when reducing the term is altogether
$\frac{1}{4}+\frac{1}{4}+\frac{1}{2}=1$. A more interesting example
is in Figure~\ref{fig:DD}, and consists in the evaluation of our
running example $\mex\mex$.

All this can be conveniently formalised by means of
multidistributions; each element corresponds to a branch, \ie to a
possible reduction path of the underlying program --- a
multidistribution is essentially a distribution on such
paths\footnote{In the spirit Markov Decision Processes, see \eg
  \cite{Puterman94}.}.  If switching to distributions, we would loose
the precise correspondence with probabilistic branches, since many
branches are collapsed into one. This is the ultimate reason why we
adopt multidistributions, and will be discussed further in
Section~\ref{sec:multi}.

Let $\MDST{\PLambdaCBV}$ denote the set of multidistributions on
(closed) terms.  We define a reduction relation $\Red\subseteq
\MDST{\PLambdaCBV}\times \MDST{\PLambdaCBV}$, given in Figure
\ref{fig:steps} and Figure \ref{fig:lifting}, respectively.  More
precisely, we proceed as follows:
\begin{varitemize}
\item
  We first  define a \emph{reduction relation $\red$}
  from terms to multidistributions, e.g.,  $M\oplus N \red
  \mdist{\two M, \two N}$. The one-step reduction $\red\subseteq
  \PLambdaCBV \times \MDST{\PLambdaCBV}$ is defined in Figure~\ref{fig:steps}.  A
  term $M$ is \emph{normal}, or in \emph{normal form}, if there is no $\m$
  such that $M\red\m$. 
   Please notice that closed terms are in normal form  precisely when they are \emph{values}. 
  Finally, $\red$ is deterministic, i.e., for every
  term $M$ there is \emph{at most one} $\m$ such that $M\red\m$.
\item
 Then we lift reduction of terms to \emph{reduction of
   multidistributions} in the natural way, obtaining $\Red $,
 e.g., $\mdist{\two II, \two (M\oplus N)} \Red $
 $ \mdist{\two I, \four M, \four N}$.  The relation
 $\red\subseteq \PLambdaCBV\times \MDST{\PLambdaCBV}$ lifts to a relation
 $\Red  \subseteq \MDST{\PLambdaCBV}\times \MDST{\PLambdaCBV}$ as defined in
 Figure~\ref{fig:lifting}.  The way $\Red $ is defined implies that
 \emph{all} reducible terms in the underlying multidistributions
 are actually reduced according to $\red$.
\end{varitemize}
\begin{figure}[t]\centering
 		\begin{minipage}[c]{0.97\textwidth}			
 			\[
 			\infer[\beta]{(\lam x.M)V\red \mdist{M\subs x V}}{}  \qquad\qquad
 			\infer[\letr V]{\lett V M  \red   \mdist{M \subs x V}}{}\\[4pt]
                        \]
                        \[
 			\infer[\oplus]{M\oplus N \red \mdist{\two M, \two N}}{} \qquad
 			\infer[\letr C]{(\lett N M)  \red \mdist{p_i (\lett {N_i} M)}}{N\red \mdist{p_iN_i}_{\iI}}
 			\]	
 	 		\end{minipage}			
 	\caption{The One-step Reduction Relation $\red$} \label{fig:steps}
\end{figure}
\begin{figure}[t]\centering
 		\begin{minipage}[c]{0.97\textwidth}			
 			\[
 			\infer[]{\mdist{V}\Red  \mdist{V}}{}\qquad \qquad
 			\infer[]{\mdist{M}\Red  \m}{M\red\m}  \qquad  \qquad 
 			\infer[]{ \mdist{p_{i}M_{i}\mid i\in I} \Red   \dsum_{\iI} {p_i\cdot \m_i}} 
 			{(\mdist{M_i} \Red   \m_i)_{\iI} }
 			\]
 		\end{minipage}	
        \caption{The Lifting of $\red$ }\label{fig:lifting}
\end{figure}

\paragraph{Reduction Sequences.}

A $\Red$-sequence (or \emph{reduction sequence}) from $\m$ is a
sequence $\m=\m_0,\m_1,\m_2,\dots$ such that $\m_{i}
\Red\m_{i+1}$ for every $i$.  Notice that since multidistribution reduction is
deterministic, each $\m_0$ has a unique maximal reduction sequence,
which is infinite and which we write $\seq \m$.  We write
$\m_0\Red^*\m$ to indicate the existence of a finite reduction sequence from $\m_0$, and $\m_0\Red^k
\m$ to specify the number $k$ of $\Red$-steps. Given a term $M$ and $\m_0=\dist M$, the sequence $\m_0 \Red
\m_1\Red\cdots$ naturally models the evaluation of $M$; each $\m_k$ expresses the
``expected'' state of the system after $k$ steps.
\begin{example}\label{ex:main2}
  The term $\mex\mex$ from Example \ref{ex:running}  evaluates
  as follows:
  \begin{align*}
    \mdist {\mex\mex}&
    \Red\mdist{\mex\mex\oplus I} 
    \Red\mdist{\two \mex\mex,\two I}\\
    &\Red 
    \mdist{\two \mex\mex\oplus I,\two I} 
    \Red\mdist{\four \mex\mex, \four I, \two I}\\
    &\Red\mdist{\four \mex\mex \oplus I, \four I, \two I}
    \Red\mdist{\frac{1}{8} \mex\mex, \frac{1}{8} I, \four I, \two I}\Red\cdots
  \end{align*}
  The first three reduction steps match precisely what we have
  informally seen in Figure~\ref{fig:DD}.
\end{example}

\newcommand{\ssym}{\mathbf{succ}}
\begin{example}\label{ex:evalCC}
  The term $\cex\cex$ of Example~\ref{ex:succ} (where $ \cex=(\lam x. \succ{xx}\oplus \nzero) $),  evaluates as follows:
  \begin{align*}
    \mdist{\cex\cex}&\Red \mdist{\succ{\cex\cex} \oplus \nzero } 
    \Red \mdist{\two \succ{\cex\cex}, \two \nzero} \\
    &	\Red
    \mdist{\two \ssym\Big({\succ{\cex\cex} \oplus \nzero} \Big), \two \nzero} 
    \Red \mdist{\four \ssym\Big({ \succ{\cex\cex}}\Big), \four  \succ{\nzero}, \two \nzero}\Red\cdots
  \end{align*}
  where $\succn n {\nzero}	~ {\Red}^{n\const}   ~\overline{n}$.
  Observe that the evaluation of $\cex\cex$ is similar to that of
  $\mex\mex$. However, while in Example~\ref{ex:main2} the term $I$ is
  a value, $\succn n {\nzero}$ is not. It still has to perform $n \const$
  steps (where $\const$ is a constant, see Example~\ref{ex:succ}) in
  order to reduce to the value $\overline {n}$.
\end{example}

\paragraph{Values and Multidistributions.}   
Given a multidistribution $\m\in \MPLambda$, we indicate by $\m^{\Val}$ its restriction
to values.  Hence if $\m=\mdist{p_iM_i}_{\iI}$, then $\nnorm \m :=
\sum_{M_i\in\Val} p_i $. The real number $\nnorm \m
$ is thus the probability that $\m$ is a value, and we will refer
to it this way. 
Looking at Example \ref{ex:main2}, observe that after, e.g., four reduction
steps, $\mex\mex$ becomes the value $I$ with probability $\two+\four$.
More generally, after $2n$ steps, $\mex\mex$ is a value with probability
$\sum_{k=1}^n \frac{1}{2^k}$.

\subsection{Probabilistic Termination in $\PLambdaCBV$}\label{sec:ASTPAST}
Let $M$ be a closed term, and $ \mdist
{M}=\m_0\Red \m_1\Red\m_2\Red\cdots $  the reduction sequence  which models its evaluation. We write
$\eval k M$ for $ \nnorm {\m_k}$, which expresses the probability that
$M$ terminates in at most $k$ steps.

\paragraph{Termination.}
Given a closed term $M$, the \emph{probability of termination of $M$}
is easily defined by $\ProbTerm{M}:=\sup_n \{\eval n {M}\}$. As an example
$\ProbTerm{\mex\mex}$ is
$\sup_n\sum_{k=1}^{n}\frac{1}{2^k}=\sum_{k=1}^\infty\frac{1}{2^k}=1$.

\paragraph{Expected Runtime.}
We now define the \emph{expected runtime} of $M$, following 
the literature~\cite{KKM2019, Fioriti, ADLY2020}.  As pointed out in
\cite{Fioriti} the expected runtime can be expressed\footnote{This because the runtime is a random variable
  taking values into $\Nat$; we therefore can easily compute its
  expectation by using the telescope formula, see \eg
  \cite{Bremaud}, page 27. The equivalence with the formulations we give below
  is spelled out in \cite{ADLY2020}.} in a very
convenient form
as, informally, 
\begin{align*}
  \ETime{M}&= \sum_{k \geq 0} \Pr[``M \mbox{ runs  more than $k$ steps }"]\\
       &=\sum_{k \geq 0} (1- \Pr[``M \mbox{ terminates within $k$ steps }"] ).
\end{align*}
Within our setting, the above is easily formalised as follows:
\[
\ETime{M} = \sum_{k \geq 0} \left(1- \eval k M\right).
\]
This formulation admits a very intuitive interpretation: given
the reduction sequence $\mdist{M}=\m_0\Red \m_1\Red\m_2\Red\cdots$, each
tick in time (\ie\ each $\Red $ step) is weighted with its probability to take place
---more precisely, the probability that a redex
is fired.  Since only (and all) terms which are \emph{not} in normal
form reduce, the system in state $\m_i$ reduces with probability 
$1- \nnorm {\m_i}$.

\paragraph{Finite Approximants.} 
Given a term $M$, the number $\eval n M$ is a finite approximant (the
\emph{$n$-th approximant}) of $\ProbTerm{M}$.  It is useful to define finite
approximants for $\ETime{M}$ too:
\[
\etime n M := \sum_{k=0}^{n-1} \left(1-  \eval k M\right).
\]
Clearly $\ETime {M} = \sup_n\left\{\etime n M\right\}$.

\begin{example}[Expected runtime, and its approximants]\label{ex:etime}
Consider again the evaluation of the term $\mex\mex$; let us decorate
each step $\m_k\Red \m_{k+1}$ with the expected probability that a
redex is actually fired in $\m_k$, that is $1- \nnorm {\m_k }$:
\begin{align*}
\mdist {\mex\mex}
&\overset {\RED 1}	\Red  \mdist{\mex\mex\oplus I} 
\overset {\RED 1}	\Red \mdist{\two \mex\mex,\two I} 
\overset {\RED \two}	\Red \mdist{\two \mex\mex\oplus I,\two I}
\overset {\RED \two}	\Red \mdist {\four \mex\mex, \four I,  \two I}\\
&\overset {\RED \four}	\Red \mdist {\four \mex\mex \oplus I, \four I,  \two I}
\overset {\RED \four}	\Red  \mdist{\frac{1}{8} \mex\mex, \frac{1}{8} I,  \four I,  \two I} \Red\cdots
\end{align*}
 It is immediate  to verify that $\ETime M=4$.
 As for the approximants, we have that, e.g.
 $\etime 2 M = 2$,  $\etime 4 M= 3$, $\etime 6 M= 3+\two$.
\end{example}

\paragraph{(Positive) Almost-Sure Termination in $\PLambdaCBV$.}
We now have all the ingredients to define the two concepts
this paper aims to characterise, namely the two canonical notions of termination.
The definition turns out to be very easy.
\begin{Def}
  Let $M$ be a closed term. We then say that
  $M$ is \emph{almost-surely terminating}   (\AST) if $\ProbTerm{M}=1$. Furthermore,
  we say that
  $M$ is \emph{expected to terminate}, or \emph{positively almost-surely
  terminating} (\PAST) if $\ETime M$ is finite.
\end{Def}
As is well-known, \PAST is strictly stronger than \AST. 
Indeed, it is easily seen that \PAST implies \AST:
\begin{fact}
  For every closed term $M$, $\ETime M<\infty$ implies that $\ProbTerm{M}=1$.
  Indeed, $\sum_{i \geq 1} (1-\eval i M) < \infty$  implies $\lim_{i\to \infty} \left(1- \eval i M\right)=0$, 
  hence $\lim_{i\to \infty} \eval i M=1$.
\end{fact}
However, a program may be \AST, and still have infinite expected runtime. The
paradigmatic example of this is a fair random walk~\cite{Billingsley}. With a slight
abuse of notation, we often use \AST and \PAST both as acronyms and as
sets of terms.

\section{Non-Idempotent Monadic Intersection Types}\label{sec:intersectiontypes}
In the previous section, we have introduced a call-by-value
paradigmatic programming language for probabilistic computation,
endowed it with an operational semantics, and defined two notions of
probabilistic termination for it. In this section, we present a type
system which is able to capture both the \emph{probability of
  termination} and the \emph{expected runtime} of a program. The
system will in turn allow us to characterise \AST and \PAST.

One of the main ingredients of the type system we are going to
introduce is the non-idempotency of intersections. As sketched in
Section ~\ref{sec:intro_IT}, such a type system is usually based on
two mutually recursive syntactic categories of types, namely simple
(or arrow) types and intersection types, which are \emph{finite
  multisets of arrows}. The intuition is that an arrow type
corresponds to a single use of a term, and that if an argument is
typed with a multiset containing $k$ arrows types, it will be
evaluated $k$ times.

In our probabilistic setting, the type system is based on
\emph{three}, rather than \emph{two}, layers, namely arrow types,
intersection types, and multidistribution types, also known
as monadic types. More precisely:
\begin{varitemize}
\item
  An \emph{arrow type} corresponds to a single use of a value, as usual.
  In a purely applicative language like ours, indeed, the only way
  to destruct a value is to pass another value to it.
\item
  An \emph{intersection type}, instead, is no longer a multiset of
  arrows like in usual non-idempotent intersection type disciplines,
  but a multiset of \emph{pairs} $q\for \At$, where $\At$ is an arrow type,
  and $q\in (0,1]\cap \Qnum$. The intuition is that each single use of
    a term will happen with some probability $q$, and that $q$ is
    recorded in the intersection type together with the corresponding
    arrow.  So, e.g., in the evaluation of $\mex\mex$ (see Example
    \ref{ex:main2}) the argument $\mex$ is first used with certitude
    (probability $1$); its next use happens with probability $\two$,
    the following use with probability $\four$, and so on.  Each use
    is typed with an (appropriately scaled) arrow type.
\item
  Finally, a term $M$ cannot in general be typed ``with certitude''
  namely by a single intersection type $\mA$, but rather with a
  \emph{multidistribution} of intersection types $\mdist{p_1\mA_1,
    \dots, p_k\mA_k}$. Indeed, the evaluation of $M$ can result
  in possibly many values depending on the probabilistic choices
  the term encounters along the way. In turn, those values can
  be copied, and each possible use of them must be taken into account.
\end{varitemize}

We now formally introduce the type system. In Section
\ref{sec:examples_CbV} we expand the intuitions above by analysing
some type derivations of our main example $\mex\mex$.  To understand
the typing, the reader should not hesitate to jump back and forth
between the examples and the formal system.
\subsection{The Type System, Formally}\label{sec:typesCBV}
\paragraph{Types}
Types are defined by means of the following grammar:
\begin{align*}
  \At,\Bt & ::=   {\ma \arrow \at}  &&\mbox{\textbf{Arrow Types}}\\
  \mA,\mB & ::=  \mul{q_1\for \At_{1},...,q_n \for\At_{n}}, n \geq 0 &&\mbox{\textbf{Intersection Types}}\\
  \at,\bt & ::=  \mdist{p_{1}\mA_{1},...,p_{n}\mA_{n}}, n \geq 0 &&\mbox{\textbf{Type Distributions}}
\end{align*}
In other words, an intersection type $\mA$ is a \emph{multiset of
  pairs} $q\for \At$ where $\At$ is an arrow type, and $q\in (0,1]\cap
  \Qnum$ is said to be a \emph{scale factor}. Note that $q>0$. Letters
  $u,q$ range over scale factors.  Given
  $\mA=\mset{q_i\for\At_i}_{\iI}$, we write $u\for \mA$ for $ \mset{(u
    q_i)\for\At_i}_{\iI}$.

It is useful to notice that an intersection type is \emph{not} a
multidistribution, because the sum of the $q_i$ such that
$\mA=\mset{q_i\for\At_i}_{\iI}$ is not bounded by $1$ in general; this
is reflected in distinct bracket notations. Intersection types
and type distributions are indeed fundamentally different.  Each
element in an intersection type corresponds to one use of the term, e.g.,
$\Delta\Delta$ can have a type of the form
$\mset{1.\At,1.\mset{\At}\arrow \at}$. Instead, type distributions are
probabilistic sums of possibly different intersection types. We often
need to multiply intersection types or type distributions by scalars,
getting other objects of the same kind. Moreover, intersection types
and type distributions being multisets, they support the (respective) operation of
disjoint union (see Section~\ref{sec:preliminariesM}).

\paragraph{Contexts.}
A \emph{typing context} $\Gamma$ is a (total) map from variables to
intersection types such that only finitely many variables are not
mapped to the empty multiset $[]$. The \emph{domain} of $\Gamma$ is
the set $\mathit{dom}(\Gamma):=\{x\st \Gamma(x)\not=[]\}$.  The typing context
$\Gamma$ is empty if $\mathit{dom}(\Gamma)=\emptyset$. Multiset union $\uplus$
is extended to typing contexts pointwise, i.e.  $(\Gamma \uplus
\Delta)(x) := \Gamma(x) \uplus \Delta(x)$, for each variable $x$.  A
typing context $ \Gamma $ is denoted as $\mset{x_1 :\mA_1, . . . , x_n
  :\mA_n}$ if $\mathit{dom}(\Gamma) \subseteq \{x_1, . . . , x_n\}$ and $
\Gamma(x_i) = \mA_i $ for all $1\leq i \leq n$. Given two typing
contexts $ \Gamma$ and $ \Delta$ such that $\mathit{dom}(\Gamma) \cap
\mathit{dom}(\Delta) = \emptyset $, the typing context $\Gamma,\Delta$ is
defined as $(\Gamma,\Delta)(x) := \Gamma(x)$ if $x \in
\mathit{dom}(\Gamma)$, $(\Gamma,\Delta)(x) := \Delta(x)$ if $x \in
\mathit{dom}(\Delta)$, and $ (\Gamma,\Delta)(x) := []$ otherwise.
Observe that $\Gamma, x :[]$ is equal to $\Gamma$. 
If 
$\Gamma = x_1:\mA_1, \dots, x_n:\mA_n$, we write  $q\for \Gamma$ for 
$[x_1:q\for\mA_1, \dots, x_n:q\for \mA_n]$.
%

\paragraph{Typing rules.}  
The type assignment system in Figure \ref{fig:WWcbv} proves judgments of
the shape $\Gamma \oder w M: \type$, where $\Gamma$ is a type context,
$M$ is a term, $w\in \Qnum $ is a weight, and $\type$ is a type in one
of the three forms, i.e.  $\type ::= \At\mid\at\mid \mA$.  If $\Pi$ is
a formal derivation of $\Gamma\oder w M:\type$, then $w$ is said to be
the \emph{weight} of $\Pi$. Please notice in Figure \ref{fig:WWcbv} the
use of the notation $q\for \Gamma$  defined above.

\begin{figure}\centering
  {\small 
    \begin{minipage}{0.97\textwidth}
        \vspace{8pt}
	\[
        \infer[\TVar]{ x: \mA  \ovdash 0 x:\mA}{}  
	\quad\quad
        \infer[\TZero]{  \oder{0 } M:\zero}{}
        \]
        \[
        \infer[@]{\Gamma  \uplus  \Delta \oder {w+ v} VW: {\bt}}
              {\Gamma \oder {w} V: [\mA\arrow \bt]   & \Delta\oder {v} W: \mA}	
              \qquad
        \infer[\oplus]{\two\for\Gamma \uplus \two\for\Delta \oder{\two w +\two v + 1} M\oplus N: \two \at \dsum \two \bt}
              {\Gamma \oder {w} M: \at & \Delta \oder{v} N: \bt} 
        \]
        \[
        \infer[\lam]{\Gamma \ovdash {w+1} \lam x. M: \mA \arrow \bt}
	      {\Gamma, x:\mA \ovdash {w} M:\bt}  
	\quad\quad        
        \infer[\letr]{\Gamma  \uplus_k  p_k\for\Delta_k \oder {v+\sum_k p_k w_k+1} \lett  N M: {\dsum_k~ p_k \bt_k}}
              {\Gamma \oder {v} N: \dist { p_k \mA_k }_{\kK} & (\Delta_k,  x:  \mA_k\oder {w_k} M: \bt_k)_{k\in K}}	
        \]
        \[
	\infer[\TVal]{\Gamma\oder {w} V:\dist{\mA} }{\Gamma\oder {w} V:\mA}
	\quad\quad
	\infer[!]{\uplus_i (q_i\for\Gamma_i)\oder {\sum_i  q_i w_i} V:\mset{q_i\for\At_i}_{\iI}}
	      {(\Gamma_i\oder {w_i} V: \At_i)_{\iI}\quad &  \quad(q_i)_{\iI} \mbox{ scale factors}
	      }
	      \]
              \vspace{0pt}
    \end{minipage}}
    \caption{Non-Idempotent Intersection Type Rules for $\PLambda^\cbv$}\label{fig:WWcbv}
\end{figure}

\subsection{Some Comments on the Typing Rules}
This section  provides some explanation on the shapes and
roles of the  typing rules. 

The leaves of type derivations seen as trees can be of two kinds,
namely the $\TVar$-rule and the $\TZero$-rule.  In both cases the
underlying weight is set to $0$. While the former is standard in
intersection type disciplines, the latter attributes the empty
distribution $\zero$ to \emph{any term} $M$. So for example,
$\Delta\Delta$ is typed as $\zero$.
The purpose of $\TZero$ is to allow for approximations, by allowing the
typing process to stop at any point. 

The next four rules are concerned with the four term constructors 
$\PLambdaCBV$ includes, namely applications, probabilistic sums, abstractions,
and $\letr$s. The following discusses each of them:
\begin{varitemize}
\item
  The $\lambda$ rule types a lambda abstraction, and assigns an arrow
  type to it. This poses no problem, because the type assigned to the
  variable $x$ in the underlying typing context is an intersection
  type, and this matches the shape of the left-hand-side of an arrow
  type. The weight is increased by one: whenever this abstraction will
  be used as a function, a $\beta$-redex would fire, and this takes
  one reduction step, which needs to be counted.
\item
  The $\oplus$ rule types $M\oplus N$ by ``superimposing'' the
  derivations for $M$ and for $N$. The data carried by each such
  derivation (context, weight and type) are scaled by a factor of
  $\two$. The counter is increased by $1$, to record a $\oplus$-step
  in the evaluation.
\item
  The $\letr$ rule serves to type the \texttt{let} construct, and is
  probably the most complex one. In particular, the argument $M$ needs
  to be typed multiple times, one for each scaled multiset in the
  multidistribution $\mdist{p_k\mA_k}$, which is the type for the
  first argument $N$. Each subderivation will be used with probability
  $p_k$, therefore in the conclusion of the $\letr$-rule, the data of
  each of the leftmost premisses (typing context $\Delta_k$, weight
  $w_k$, and type $\bt_k$) are scaled by a factor $p_k$. Moreover, the
  weight is further increased by $1$, to record a $\letr V$-step in
  the evaluation; such a step consumes the $\letr$ when the first
  argument is a value.
\item
  Finally, the $@$ rule typing applications is quite standard in
  shape. Just a couple of observations could be helpful. First of all,
  the function $V$ is required to be typed with a multiset, rather
  than a distribution, and this is not restrictive since $V$ is a
  value, and not a term. Secondly, the weight is taken as the sum of
  the weights of the two derivations, without any increase. Notice
  that the corresponding $\beta$-step is recorded by the
  $\lambda$-rule.
\end{varitemize}

The last two rules, namely $\TVal$ and $!$, are the only ones not
associated to any term construction, and are meant to allow a term
typable with arrow types to be attributed an intersection or
distribution type. Of course, this makes sense only when the term is
actually a value.

\begin{remark}
  While designing the type system, we made a simplifying choice in the
  typing rule $\letr$. As we said, the counter is increased by $1$ to
  record the $\letr V$-step.  Note however that if $N$ is a
  non-terminating term, it never becomes a value, and therefore the
  $\letr V$-step never happens. Are we counting\emph{ too much} here?
  Obviously not, because if $N$ never become a value, then any
  reduction sequence from $\lett M N$ can be extended with an extra
  reduction step, without affecting the analysis in any way.
\end{remark}

\paragraph{Some Interesting Boundary Cases.}
The type system we have just introduced is remarkably simple in structure,
despite its expressive power, which we will analyse  in
Section~\ref{sec:characterization}. Let us now take a look
at a few degenerate cases of the typing rules:
\begin{enumerate}
\item
  In the $\lam$ rule, $\mA$ is allowed to be the \emph{empty intersection}
  type, this way allowing to type vacuous abstractions, i.e., we can
  always abstract a variable $x$ which does not explicitly occur in
  the context $\Gamma$, since if $x \not\in dom (\Gamma)$, then
  $\Gamma, x :[]$ is equal to $\Gamma$.
\item
  In the $!$-rule, $I$ can be empty, and the following rule
  is thus a derived rule:
  \[
  \infer{\oder{0} \lam x.M:[]}{}
  \]
\item
  In the \texttt{let} rule, the term $N$ can well have \emph{null type} $\zero$, and in this case
  the whole term $\lett N M$ is given itself type $\zero$, without any need to
  type $M$. In other words, the following is another derived rule
  \[
  \infer{\Gamma \oder {w+1} \lett N M:\zero}{\Gamma \oder{w} N:\zero }				
  \]		
\end{enumerate}

\subsection{Some Basic Properties of the Type System}
In this section, we derive some easy but useful properties of the type system,
which will turn out to be essential in the following. Like in linear type systems,
typing contexts tell us everything we need to know about free variables:
\begin{lemma}[Contexts and Free Variables]
 Let $\Gamma \der M: \at$.  Then $\mathit{dom}(\Gamma)\subseteq fv(M)$, and  $M$  closed implies $\Gamma=\emptyset$.

\end{lemma}
The way intersection types are assigned to values is  completely determined
by the underlying arrow types: 
\begin{property}[Partitioning Intersections]
  For every value $V$, the following are equivalent:
  \begin{enumerate}
  \item
    $ \oder {w} V: \mA$ and $\mA= \uplus_{\iI} \mA_i$;
  \item
    $\oder{w_i} V: \mA_i$ for every $\iI$ and $w=\sum w_i$.
    \end{enumerate}
\end{property}
We often use the aforementioned property together with the following lemma:
\begin{lemma}[Scaling]\label{lem:scale} Given any scalar $0<q \leq 1$ and any value $V$, it holds that
  $ \oder{qv} V:q \for \mA $ ~iff~ $ \oder{ v} V:
    \mA $.
\end{lemma}


\section{Precisely Reflecting the Runtime: Some Examples}\label{sec:examples_CbV}

Our type system is designed to keep track of the probability of
reaching a value and the expected time to termination.
And as it should by now be clear, the information relevant to derive the latter is kept
track by the weight.  Since the expected runtime is computed as an
infinitary sum, working with exact measures is essential.
Think for example at $\sum_{k=1}^\infty \frac{1}{k^2}$
and $\sum_{k=1}^\infty \frac{1}{k}$: the first converges, while the second
diverges.

We thus need to count steps neither ``too much'' nor ``too little''.  Two
crucial features make this possible: the arrows in an intersection
type are \emph{scaled} by a factor $q$, and we allow type derivations
also for terms which receive the \emph{null type} $\zero$, such as
$\Delta\Delta$. The first feature allows us not to count ``too much'',
the second not to count ``too little''. In this section, we illustrate
these aspects by way of some concrete examples.
\subsection{Not Too Much}\label{sec:mainex_CbV}
In intersection type systems for the $\lam$-calculus such as those by
Lengrand and co-authors \cite{BL2013,AGK2018,AGL19}, the weight of any
type derivation accounts for how many times redexes can be
fired. Roughly speaking, to each $\lambda$-abstraction in the type
derivation corresponds a $\beta$-redex being fired, therefore to
measure the runtime of a $\lam$-term, the weight is increased by one
at each instance of the $\lambda$ rule. The only difficulty consists
in distinguishing between those abstractions which are used as
functions, and those abstractions which will turn out to be
the final value.

In a probabilistic setting, we want to compute the \emph{expected}
runtime. Increasing the weight by one \emph{at each} instance of the
$\lambda$ rule as in the deterministic case is simply too much.
Consider our running example $\mex\mex$, where $\mex=\lam x .xx\oplus
I$, and $\ETime {\mex\mex}=4$.  It is easy to see that for each $k\in
\Nat$, there is a derivation $\Pi$ which contains $k$ instances of
$\lam$ rule. If each is counted $1$, we would have a derivation $\Pi
\oder k \mex\mex: \at$ for every $k$, and so $\sup\{ k \mid \oder k
\mex\mex: \at \}=\infty$.  Instead, we need to scale each instance of
$\lam$ (say, with conclusion $\At$) by \emph{the probability $p$ of
  the $\lambda$ abstraction to be involved in a redex}. Such an
information is stored as a scalar somewhere else in the derivation.

To clarify, let us examine our running example.  We want to
capture $\etime {n} {\mex\mex} $ and the fact that $\eval {2n}
{\mex\mex}$ is $\left(\frac{1}{2}+\frac{1}{4}+ \dots +\frac{1}{2^n}\right)$.
We define  the types $\At_n$ and $\ma_n$ as follows. 
 \[
\ma_0= []
\qquad
\At_n=\ma_{n-1}\arrow  \dsum_{k=1}^n \mdist{\frac{1}{2^k}[]}
\qquad
\ma_n =  \two\for \ma_{n-1} 	\uplus 	\two\for  [\At_n]
\]
For the reader's convenience, we explicitly give some cases:
\begin{align*}
\At_1&=[] \arrow \dist{ \two []},\quad &
\At_2&=\ma_1\arrow  \dist{\frac{1}{4}[] , \frac{1}{2}[]},\quad&
\At_3&=\ma_2\arrow   \dist{\frac{1}{8}[], \frac{1}{4}[], \frac{1}{2}[]}\\
\ma_1&=\mset{\two \for\At_1},&\quad\quad  \ma_2& =\mset{\frac{1}{4} \for\At_1,  \frac{1}{2}\for\At_2},\quad&
\ma_3&=\mset{\frac{1}{8}\for\At_1, \frac{1}{4}\for\At_2, \frac{1}{2}\for\At_3}
\end{align*}
The value $\mex$ can be given all the arrow types $\At_i$, for every
$i$, all these derivations having weight equal to $2$, i.e. for every
$i\geq 1$ there is a derivation $\Sigma_i$ such that
$\Sigma_i\dem\ovdash{2}\lam x. xx \oplus I: \At_i$. For example, the
type derivations $\Sigma_1$ to $\Sigma_3$ can be built as follows:
\[\Sigma_1\dem\infer{\ovdash{2} \lam x. xx \oplus I:
    []\arrow \dist{ \two[]}} {\infer{\ovdash{1} xx\oplus I:
      \dist{\two []}}{\ovdash 0 I: \dist{[]}} }
  \qquad\quad 
\Sigma_2\dem
\infer{\ovdash{2}  \lam x. xx \oplus I:  [\two\for\At_1]\arrow  \dist{\frac{1}{4}[], \frac{1}{2}[]}   } 
{\infer{x:[\two\for\At_1] \ovdash {1} xx \oplus I:  \dist{\frac{1}{4}[], \frac{1}{2}[]}}
	{ \infer{x:[\At_1]  \ovdash  0 xx: \dist{\two []}}
		{ x:[\At_1]\ovdash  0 x: [[]\arrow \dist{\two []}]  & x: []\ovdash  0 x:[] }
		& \ovdash  0 I:\dist{[]} }
}
\]
\vspace{5pt}
\[
\Sigma_3\dem
\infer{\ovdash{2} \lam x.  xx\oplus I: [\two \for\At_2, \frac{1}{4}\for \At_1]\arrow \dist{ \frac{1}{8}[], \frac{1}{4}[], \two []  }  }{
			\infer{x:[\two \for\At_2, \frac{1}{4}\for \At_1] \ovdash{1} xx\oplus I:  
				\dist{ \frac{1}{8}[], \frac{1}{4}[], \two [] } }{
				\infer{x:[ \At_2, \two \for\At_1] \ovdash  0 xx:    \dist{\frac{1}{4}[], \frac{1}{2}[]}  }{
					x:[ \At_2] \ovdash 0 x:  [ [\two\for \At_1]\arrow     \dist{\frac{1}{4}[], \frac{1}{2}[] }]      
					& {x:[\two\for \At_1]\ovdash 0 x:[\two \for\At] }
				}
			&
				\infer{\ovdash  0 I: \dist{[]}}{}
			}
		}
\]
One can  attribute  to $\mex$ also any   
intersection type $\ma_j$, by collecting and scaling the $j$ derivations  ($\Sigma_j,\dots, \Sigma_1$) by way of the rule $!$ (with scale factors $\two, \dots,\frac{1}{2^j}$), 
thus obtaining the
type derivation $\Theta_j$,  this time with weight $2(\sum_{k:1}^j \frac{1}{2^k}) = 2-\frac{1}{2^{j-1}}$.
 Finally, $\Sigma_{j+1}$ and $\Theta_j$ can
be aggregated in the derivation $\Phi_{j+1}$, typing $\mex\mex$. Note that the weight is now $2(\sum_{k:0}^j \frac{1}{2^k}) $.
\[\Phi_{j+1}\dem 
\infer{
	\ovdash { 2(\sum_{k:0}^j \frac{1}{2^k})    } \mex\mex: \dsum_{k:1}^{j+1} \dist{\frac{1}{2^k}[]}}
{\Sigma_{j+1}\dem\ovdash{2}\mex:  \ma_j\arrow \dsum_{k:1}^{j+1} \dist{\frac{1}{2^k}[]}
	&
	\Theta_j\dem\ovdash{2(\sum_{k:1}^j \frac{1}{2^k})  } \mex: \ma_j}
\]

\begin{sloppypar}
For example,  we  have the following derivation, which indeed corresponds to the  $(2\cdot 3)$-approximant of $\ETime {\mex\mex}$:
 recall from Example~\ref{ex:etime} that $\etime 6 {\mex\mex}=3+\two$ (and that $\eval 6 {\mex\mex} = \frac{7}{8}$). 
\end{sloppypar}
	\[ \Phi_3\dem
	\infer{ \ovdash {3+\two} \mex\mex:  \dist{\frac{1}{8}[], \frac{1}{4}[], \two [] }}
	{\Sigma_3\dem\ovdash{2} \lam x.  xx\oplus I: [\two \for\At_2, \frac{1}{4}\for \At_1]\arrow \dist{ \frac{1}{8}[], \frac{1}{4}[], \two []  }  
	&
	 \infer[!]{\ovdash{1+\two} \mex:  [\frac{1}{2}\for \At_2, \frac{1}{4}\for\At_1]}
	 {  \Sigma_2 \ovdash{2} \mex:\At_2  & \Sigma_1 \dem  \ovdash {2} \mex:\At_1}    }
	\]

\subsection{Not Too Little}\label{sec:too_little}
Our type system allows to count the reduction steps of diverging
terms. That is, a term such as $\Delta \Delta$ has a derivation of
weight $n$, for each $n\in \Nat$.  This is essential to precisely
capture the expected runtime. Think of the term $M:=I\oplus
\Delta\Delta $. Its evaluation proceeds as follows:
 \[
\mdist{I\oplus \Delta\Delta}
\overset 1 \Red \mdist{\two I, \two \Delta \Delta} 
\overset \two  \Red    \mdist{\two I, \two \Delta \Delta} 
\overset \two    \Red \mdist{\two I, \two \Delta \Delta}   \dots
\]
Clearly, $\ETime{M}=\infty$. However, any derivation only taking into
account the evaluation time \emph{to a value} (namely the $\oplus$
reduction step only), would necessarily have finite weight. In the
following, we prove that any diverging program $M$ can be typed as
$\oder n M: \zero$,  for \emph{every} natural number $n$.

Here, we show this fact, concretely, for the paradigmatic diverging term $\Delta\Delta$.
First of all, consider the arrow types $\At_{i+1}=[\At_1,\ldots,\At_{i}]\to \zero$
(so, in particular, $\At_1=[]\to \zero$, $\At_2=[\At_1]\to \zero$). For each $i$, one can
build a derivation $\Sigma_i$ having weight $1$ and
typing $\Delta$ with $\At_i$. Here are a couple of examples:
\[
  \Sigma_1\dem\infer{\oder{1}\lam x.xx : []\arrow \zero}{\oder 0 xx:\zero}
  \qquad
  \Sigma_2\dem\infer{\oder 1 \lam x.xx : [[]\arrow \zero]\arrow \zero}{\infer{x:[[]\arrow \zero ]\oder 0 xx:\zero}
    {x: [[]\arrow \zero]\oder 0 x:[[]\arrow \zero]   & \oder 0 x:[]  } }
  \]
From the $\Sigma_i$'s, it is thus easy to build derivations
typing $\Delta\Delta$ with $\zero$ and having any weight $n$.
As an example, if $n=3$, we have the following one:

{\footnotesize
\[
\infer[@]{\oder{3}(\lam x.xx) \lam x.xx : \zero}
{\Sigma_3\dem\oder{1}\lam x. xx: \mset{ ~[[]\arrow \zero]\arrow \zero,~  []\arrow \zero ~ }\arrow\zero 
	\quad & \quad \Sigma_2\dem\oder{1}\lam x.xx : [[]\arrow \zero]\arrow \zero 
	\quad & \quad \Sigma_1\dem\oder{1} \lam x.xx  : []\arrow \zero
}
\]}


\section{Characterising Probabilistic Termination}\label{sec:characterization}
This section presents the main result of this paper, namely the
characterisation of both forms of probabilistic termination by
typing. This will be done by relating type derivations for a program
$M$ and the probability of termination and the expected runtime of
$M$. To achieve the latter, we need to focus on \emph{tight}
derivations, since not all type derivations of $M$ underapproximate
the expected runtime of $M$.

We show that in the tight case, $\ETime M$ (respectively, $\ProbTerm{M}$)
bounds from above the weight $w$ (respectively, the norm
$\norm{\at}$) of any type derivation $\Pi\dem\oder{w}{M:\at}$. This is
the \emph{soundness property}, and is in Section~\ref{sec:sound}. We also
prove the converse, \ie the \emph{completeness property}, in Section~\ref{sec:complete}.
\subsection{Tight Typings}
The need for tight typings can be grasped easily by considering
the following example.
\begin{example}\label{ex:tight}
  The term $I$ is in normal form, and therefore $\ETime I=0$. It
  can be given the type $\zero$ (by way of the $\TZero$ typing rule),
  or the type $\mdist{1[]}$ (by way of  $\, !$ and $\TVal$).
  In both cases, the underlying weight is $0$.
  However, $I$ also admits derivations whose weight is
  strictly positive, such as
  \[
  \infer[\lam]{\oder 1 \lam x.x: []\arrow []}{\infer[\TVar]{x:[] \oder 0 x:[]  }{}}
  \]
  More generally, without any restrictions on the shape of types,
  one can easily assign grossly overapproximated weights to
  terms, e.g., the term $\lambda x.(\Delta\Delta)$, which is a value
  but which can receive arbitrarily large weights when given the
  type $[]\arrow\zero$, (immediate consequence of the example in  Section \ref{sec:too_little} above).
\end{example}
The purpose of arrow types is to give types to terms which are
\emph{not} supposed to be reduced alone, but only when applied
to an argument. If, instead, a term is not supposed to be used
as a function, its type must be the empty multiset. This is the
key idea for understanding the following definition:

\begin{Def}[Tight Types and Derivations]
  A type $\at$ is said to be \emph{tight} if it is a multidistribution
  on the empty intersection type $[]$. Accordingly, a derivation
  $\Pi\dem ~ \oder w M: \at $ is said to be \emph{tight} if $\at$ is tight.
\end{Def}
A tight type has therefore shape $\at=\mdist{q_k[]}_{\kK}$, where $K$ is
possibly empty. In particular the null type $\zero$ is a tight type. Observe that if
$\at$ is tight, then $\norm {\at}=\sum_k q_k$ (which, again, is null when $K$
is empty). The following can be proved by quickly inspecting the typing
rules:
\begin{lemma}[Tight Typings for Values] \label{lem:normalV}
  If $V$ is a closed value, then there
  are precisely two tight derivations for $V$, both of weight $0$:
\[
\infer[\TVal]{\oder 0 V:\mdist{1[]}}{\infer[!]{\oder 0 V:[]}{}} \quad\quad  \infer[\TZero]{\oder 0 V:\zero}{}
\]
\end{lemma}
Looking back at Example \ref{ex:tight}, one immediately realises that
tightness allows us to get rid of overapproximations, at least for values. Does
this lift  to all terms? The next two subsections will give a positive
answer to this question. The following  property, which is immediate from the definitions,
will be useful in the rest of this section.
\begin{property}\label{fact:approximants}
  For any closed term $M$ and any  $k\in \Nat$, it holds that
$\eval k M \leq \eval {k+1} M$ and $\etime k M \leq \etime {k+1} M$. Moreover, if $M\red \mdist {q_iM_i}_{\iI}$ then
\[
\eval {k+1} M = \sum_{\iI} q_i\left(\eval {k} {M_i}\right),\qquad
\etime {k+1} M = 1+ \sum_{\iI} q_i\left(\etime{k} {M_i}\right).
\]
\end{property}

\subsection{Soundness}\label{sec:sound}
In this section, we prove the correctness of our type system. Namely,
we prove that if $\oder w M: \at$ is (tightly) derivable, then $M$ has
probability of termination at least $\norm \at$, and expected runtime
at least $w$.
  
The proof of correctness is based on the following, namely a form of
weighted subject reduction, that for good reasons has a probabilistic
flavor here.  The size of a type derivation $\Pi$ (denoted $\size
\Pi$) is the standard one, and is defined as the number of rules in
$\Pi$ (excluding the !-rule and the $\TVal$-rule, which cannot be
iterated).

\begin{lemma}[Weighted Subject Reduction]\label{lem:SubRed}\label{lem:WSR}
   Suppose that $\Pi\dem \oder{w} P:\bt$, with $w >0$, and that $P \red
  \mdist{q_iP_i}_{\iI}$. Then for every $\iI$ there exists a
  derivation $\Pi_i$ such that $\Pi_{i}\dem\oder {w_i} P_i: \bt_{i}$,
  and $\size{\Pi} > \size{\Pi_i}$. Moreover,
  $\bt= \dsum_{i\in I} q_i \bt_{i }$ and 
  $w=1+ \sum_{i\in I} q_i w_i$.		
\end{lemma}
The proof is in the Appendix.
Notice how the type stays \emph{the same}, at least on the average,
while the weight strictly \emph{decreases}. This in turn implies that whenever
a term is (tightly) typable, its weight is a lower bound to its expected
time to termination, while the norm of its type is a lower bound to the
probability of termination. This is proved by way of approximations, as
follows:
\begin{theorem}[Finitary Soundness ]\label{th:correct}
  Let $M$ be a \emph{closed} term.  For each tight typing $\oder w M:
  \bt$, there exists $k\in \Nat$ such that
  $ \norm \bt \leq \eval k M$  and 
    $w \leq \etime k M $.
\end{theorem}
\begin{proof}
  By induction on the size $\size{\Pi}$ of the type derivation $\Pi$
  such that $\Pi\dem\oder w M$, distinguishing some cases. Recall that
  for closed terms, the normal forms are exactly the values.
  \begin{varitemize}
  \item
    If $M$ is a value, the claim holds by Lemma \ref{lem:normalV}, where we observe that  $w=0$. 
    Notice that  $\eval 0 M =\nnorm {\mdist{M}}=1$ and   
    $\etime 0 M=0$.
  \item
   Otherwise, if $M$ is not a value,  we further distinguish some cases:
    \begin{varitemize}
    \item If $w=0$, then by inspecting the rules, we see that the only
      derivable tight judgment is $\oder 0 M:\zero$ which trivially
      satisfies the claim, with $k=0$.
    \item If $w>0$, then since $M$ is not normal, it has a reduction
      step $M \red \mdist{q_iM_i}_{\iI}$.  By Weighted Subject
      Reduction (Lemma \ref{lem:WSR}), we derive that for each $\iI$
      there exists a derivation $\Pi_i \dem \oder{w_i} M_i:\bt_i$,
      with $\size {\Pi_i} <\size {\Pi}$. Since $\bt$ is tight,
      necessarily each $\bt_i$ also is tight, again by
      Lemma~\ref{lem:WSR} (observe also that, if $\bt=\zero$, then
      $\bt_i=\zero$).  By \ih, for each $M_i (\iI)$ there exists
      $k_i\in \Nat$ which satisfies the conditions on $\norm{\bt_i}$
      and $w_i$.  Let $h=max \{k_i\}_{\iI}$.  Since ${h}\geq k_i$, for
      each $\iI$ we have $\norm {\bt_i} \leq \eval {h} {M_i}$ and $w_i
      \leq \etime{h} {M_i}$.  Moreover, Weighted Subject Reduction
      implies also that $ \norm \bt =\sum_i q_i \norm {\bt_i}$ and $w=
      1 + \sum_i q_i w_i$.  The claim follows easily by
      Property~\ref{fact:approximants}, with $k=h+1$. Indeed $\eval
      {h+1} M = \big(\sum_{\iI} q_i \eval {h} {M_i}\big) \geq \big(
      \sum_{\iI} q_i \norm{\bt_i}\big)= \norm{\bt}$ and $\etime {h+1}
      M = \big(1+ \sum_{\iI} q_i \etime {h} {M_i}\big) \geq \big(1+
      \sum_{\iI} q_i w_i \big)= w$.
    \end{varitemize}
  \end{varitemize}
  Since there are no other cases, we are done.
\end{proof}
Observe that Theorem~\ref{th:correct} holds for every tight type $\bt$, including the null type. Thus it has 
the following immediate consequence:
\begin{cor}[Finitary Soundness of Null Typing]\label{cor:zero_correct}
  Let $M$ be a \emph{closed} term such that $\oder w M:
  \zero$. Then there exists $k\in \Nat$ such that $w \leq \etime k M$.
\end{cor}

\subsection{Completeness}\label{sec:complete}
The last section showed that type derivations provide lower bounds on
the probability of convergence, and on the expected time to
termination. It is now time to prove that \emph{tight} derivations
approximate \emph{with arbitrary precision} the aforementioned
quantities.  The proof of completeness is based on the following
probabilistic adaptation of Subject Expansion.
\begin{lemma}[Weighted Subject Expansion]\label{lem:SubEx}
  Let $P$ be a closed term.  Assume that $P\red \mdist{q_i P_i}_{\iI}$ and
  that for each $\iI$, $\Pi_i \dem \oder {w_i} P_i:\at_i $. Then, there exists a
  single derivation $\Pi\dem \oder {w} P:\at $ such that $\at=\dsum_i
  (q_i\at_i)$ and $w\geq 1+\sum q_i w_i$.
\end{lemma}
The proof is in the Appendix.
We can now thus prove the dual to Theorem~\ref{th:correct} above:
\begin{theorem}[Finitary Completeness]\label{th:complete}
  Let $M$ be a closed term.  For each $k\in \Nat$ there exists a
  \emph{tight} derivation $\Pi\dem ~ \oder w M:\at$, such that $\norm
  \at = \eval k M$ and $w \geq \etime k M $.
\end{theorem}
\begin{proof}
	By induction on $ k$, distinguishing some cases.
        \begin{varitemize}
	\item  If $M$ is a value, then for each $k$,  $\eval k M=1$ and $\etime k M =0$.   The  derivation 
	$\oder 0 M:\mdist{1[]}$ ( Lemma~\ref{lem:normalV}) satisfies the claim. 
	\item Otherwise, if $M$ is not a value: 
	\begin{varitemize}
		\item If $k=0$, we have  $\eval 0 M=0 = \etime 0 M$; the 
		 $\TZero$-rule satisfies the claim.

		\item
		If $k>0$, assume $M\red \mdist{p_i M_i}_{\iI}$.  By \ih, for each $i\in I$, there
		exists a tight derivation $\Pi_i\dem ~ \oder {w_i} {M_i}: \at_i$, such
		that $\norm {\at_i}=\eval {k-1}{M_i}$ and $w_i \geq  \etime {k-1}{M_i}
		$. By Weighted Subject Expansion, there exists a tight derivation
		$\Phi\dem \oder {w} M: \at $ such that $\norm {\at}=  \sum_{\iI} p_i \norm {\at_i} $, and $w\geq  \sum_{\iI}p_iw_i$. We conclude  by Property~\ref{fact:approximants}, because 
	 	  $\eval k M =  \sum_{\iI}	p_i	\eval {k-1} {M_i}  = \sum_{\iI}	p_i	\norm{\at_i}=\norm \at$
	 and 
	  $\etime k M = 1+ \sum_{\iI}	p_i	\etime {k-1} {M_i}  \leq 1+ \sum_{\iI}	p_i	w_i \leq  w$.
	\end{varitemize}
\end{varitemize}
\end{proof}

\subsection{The Various Flavours of a Correspondence}
\newcommand{\W}[1]{\mathbf{W}(#1)}
\newcommand{\N}[1]{\mathbf{N}(#1)}
This section is devoted to characterisation results
relating typing and termination. The latter can be given in three
different ways, and we devote a subsection to each of them.

\subsubsection{A Uniform Characterization}
A characterization of both $\ProbTerm M$ and $\ETime M$
by \emph{the same class} of derivations, namely tight
derivations, can be given as follows:
\begin{thm}[Tight Typing and Termination]\label{th:term_tight}\label{thm:characterization}
  Let $ M$ be a closed term.  Then
  \begin{align*}
  \ProbTerm M &= \sup  \{ \norm {\at}\st \exists\Pi.\Pi\dem  \vdash M:\at \mbox{ is a tight derivation} \}\\
  \ETime M  &=  \sup \{w  \st\exists\Pi.\exists\at.\Pi\dem  \oder w M: \at \mbox{ is  tight derivation} \}
  \end{align*}
\end{thm}
\begin{proof} 
  Let us first of all define $\N{M}$ and $\W{M}$ as follows:
  \begin{align*}
    \N{M}&\deff \sup \{ \norm {\at}\st\exists\Pi.\Pi\dem  \vdash M:\at \mbox{ is a tight derivation} \}\\
    \W{M}&\deff \sup \{w  \st\exists\Pi.\exists\at.\Pi\dem  \oder w M: \at \mbox{ is  tight derivation} \}
  \end{align*}
  We now proceed by proving the following two statements:
  \begin{varitemize}
  \item
    On the one hand, we prove that $ \N M = \ProbTerm M $, recalling
    that we defined $\ProbTerm M$ as $\sup\{\eval n M\st n\in\Nat \}
    $.
    \begin{varitemize}
    \item
      Let $ \Phi\dem \oder {} M: \at$ be a tight derivation. By
      Correctness (Theorem \ref{th:correct}), $\norm \at \leq
      \ProbTerm M$. Hence $\N M\leq \ProbTerm M$.
    \item
      By Finitary Completeness (Theorem \ref{th:complete}), for each $k\in
      \Nat$ there exists a tight derivation $\Pi \dem \oder {} M: \at
      $ such that $\eval k M \leq  \norm \at$.  Hence $\N M= \ProbTerm M$.
    \end{varitemize}
  \item
    On the other hand, we prove that
    $\W M = \ETime M $, recalling that $\ETime M$ is defined as $\sup\{ \etime n M\st n\in\Nat \}$.
    \begin{varitemize}
    \item
      Let $w $ be the weight of a tight derivation $ \Phi\dem \oder w
      M: \at$. By Correctness (Theorem \ref{th:correct}), $ w\leq
      \ETime M $. Hence $\W M\leq \ETime M$.
    \item
      Again by Finitary Completeness (Theorem \ref{th:complete}), for each $k\in
      \Nat$ there exists a tight derivation $\Pi \dem \oder w M: \at $
      such that $\etime k M \leq w$.  Hence $\W M= \ETime M$.
    \end{varitemize}
  \end{varitemize}
\end{proof}

By definition, $M$ is \AST iff $\ProbTerm{M}=1$, while $M$ is \PAST iff $ \ETime M$ is finite.
As a consequence:
\begin{cor}[Tight Typing, \AST, and \PAST] \label{cor:char_tight}
  Let $ M$ be a closed term. Then:
    $M$ is \AST iff $ \sup \{ \norm {\at}\st\exists\Pi.\Pi\dem \vdash M:\at
    \mbox{ is a tight derivation} \} =1 $. Moreover, 
    $M$ is \PAST iff $ \sup \{w \st\exists\Pi.\exists\at.\Pi\dem \oder w M: \at \mbox{ is
    tight derivation} \} < \infty $.
\end{cor}

\newcommand{\ww}{\mathbf{w}} \newcommand{\oh}{\mathbf{v}}
\begin{example}

  In Section~\ref{sec:examples_CbV} we have discussed tight
  derivations for our running example $\mex\mex$. Each derivation
  $\Sigma_i$ for $\mex$ has constant weight
  $2$. Collecting and scaling the  $j+1$ derivations $\Sigma_{j+1},  \Sigma_j,\ldots,
  \Sigma_1$,  we obtained a derivation  $\Phi_{j+1}$ for $\mex\mex$, with  weight 
  $2(1+\two +... + \frac{1}{2^{j}})= 2(\sum_{n\leq j} \frac{1}{2^n}) <4$.  Let us 
  also sketch how type derivations can be built for the terms
  $\cex\cex$ and $\expl{\cex\cex}$ from Example \ref{ex:succ}.
    \begin{varitemize}
  \item
     We can build a tight type derivation $\Xi_i$ for $\cex=
     \Big(\lam x.  ( \plett z {xx} {SUCC\, z}) \oplus \nzero \Big)$ by
     following the blueprint of the derivation $\Sigma_i$ for $
     \mex$. The weight of each  $\Xi_i$ is  $2+ \two v_{i}$, where the weight 
      $ v_i$   is contributed by the  $\letr$ rule, and  increases as $i$ increases, because the $\letr$ rule has more than one premiss.
     By  collecting and appropriately scaling the  derivations
     $\Xi_{j},\dots, \Xi_1$, and putting this together with $\Xi_{j+1}$, we then obtain a type derivation $\Psi_{j+1}$
     for $\cex\cex$ (similarly to what we have done to obtain
     $\Phi_{j+1}$). The weight has now  a bound similar to that for
     $\Phi_{j+1}$, plus an overhead which is obtained by summing the  scaled $ \two v_i$'s, 
     giving an overall weight  $w< 4+ \const$.     
   \item
     An even more interesting term is $\expl{\cex\cex}$.  We can build
     tight derivations for it from appropriate derivations for $\cex\cex$.
     It is clear that the term $\EXPL~\overline n$
     can be given a tight derivation of weight (at least)
     $2^n$, in a standard way.
     From there, for every $j$, a tight type derivation having weight
     at least $\two \sum_{n\leq j} \frac{2^n}{2^n}=\frac{j+1}{2}$ can be built,
     so the set of tight weights is unbounded.
\end{varitemize}
\end{example}

\subsubsection{Focusing on Expected Runtimes.}
Remember that $\PAST\subset\AST$.  The previous characterisation may
give the impression that analysing the runtime of a term somehow
requires studying its probability of termination.  In fact, intersection
types allow us to establish \PAST \emph{independently} from \AST, by
looking only at the type $\zero$ rather than at \emph{all} tight
typings. The results we are going to prove tell us that if we are only interested in
the expected runtime, we can indeed \emph{limit the search space}
to the derivations of the null type $\zero$.  First of all,
a strengthening of Finitary Completeness can be given.
 \begin{prop}[Finitary Completeness of Null Typing]\label{th:zero_complete}
   Let  $M$ be a closed term. 
   Then, for each $k\in \Nat$ there exists a derivation $\Pi\dem ~
   \oder w M:\zero$, such that $w \geq \etime k M $.
 \end{prop}
 \begin{proof} 
   The proof is a simplification of Theorem \ref{th:complete}. We only
   need to observe  that the $\zero$ typing is preserved by subject
   expansion, and the weight strictly increases along it.
   
   \SLV{}{As
   	before, we reason by induction on $ k$.
   	\begin{varitemize}
   		\item $k=0$. We have   $\etime 0 M=0$, and the $\TZero$-rule satisfies the claim.
   		\item $k>0$. Assume  $M\red \mdist{p_i M_i}_{\iI} \full^{k-1} \m$.  		
   		For concreteness (but w.l.o.g.), let us discuss the   instance $I=\{1,2\}$,
   		$p_i=\two$.
   		From   $M\red \mdist{\two M_1, \two M_2} \full^{k-1} \m$ we have that  $\m=\m_1\uplus\m_2$, with $\mdist{M_i}\full^{k-1}\m_i$. 	By \ih, for each  $i\in \{1,2\}$, 
   		there exists a  derivation   $\Pi_i\dem ~ \oder w_i M_i: \zero$, such that  $w_i = \etime {k-1}  {M_i} $. By
   		Weighted Subject Expansion,  there exists a   derivation  $\Phi\dem  \oder {w} M: \zero$ such that 
   		$w= \two w_1+ \two w_2+1$. This proves the claim,  because  $\etime k M = 1+ \two \etime {k-1} {M_1}  +  \two \etime {k-1} {M_2} = 1+  \two w_1+ \two w_2$.
   	\end{varitemize}
   }
 \end{proof}
 Since Finitary Soundness holds \emph{at all types}, we can easily reach
 the following:
\begin{thm}[Null Typing, Expected Runtimes, and \PAST]\label{thm:char_null}
  Let $ M$ be a closed term.  Then:
  $$
  \ETime M  =   \sup \{w  \st \Pi\dem  \oder w M: \zero \}\qquad
  M\in\PAST\Leftrightarrow\sup \{w  \st \Pi\dem  \oder w M: \zero \} < \infty
  $$
\end{thm}
\paragraph{The Running Example, Revisited}
Let us go back to our running example $\mex\mex$, and show that its runtime
can be analysed by way of null types. We can indeed build type
derivations of the form $\Pi \dem \oder w \mex\mex:\zero$ in such a
way that $w$ is bounded by $\ETime {\mex\mex}=4$, and for each
approximant $\etime k {\mex\mex}$ there is a derivation which has at
least that weight.  The structure of these type derivations are
identical to the ones we gave in Section \ref{sec:examples_CbV}.  The
only difference is in how the underlying \emph{types} are defined.  Let us define
the families of types $\{\Bt_n\}_{n\in\Nat}$ and
$\{\mb_n\}_{n\in\Nat}$ as follows:
\begin{align*}
\mb_0= []\qquad\qquad
\mb_n =  \two\for \mb_{n-1} 	\uplus 	\two\for  [\Bt_n]\qquad\qquad
\Bt_n=   \mb_{n-1}\arrow  \zero
\end{align*}
For example:
\begin{align*}
\Bt_1&=[] \arrow \zero, & \Bt_2&=\mb_1\arrow  \zero, & \Bt_3&=\mb_2\arrow   \zero, \\
\mb_1 &=\mset{\two \for\Bt_1}&\mb_2 &= \mset{\frac{1}{4} \for\Bt_1,  \frac{1}{2}\for\Bt_2}, &
\mb_3 &=\mset{\frac{1}{8}\for\Bt_1, \frac{1}{4}\for\Bt_2, \frac{1}{2}\for\Bt_3}
\end{align*}
The given types are structurally very similar to those from
Section~\ref{sec:examples_CbV}. We can thus mimic the constructions
given there, and get derivations $\Sigma_i$, each having weight $2$ and
typing $\mex$ with $\Bt_i$, but also derivations having weights
converging to $4$, this time typing $\mex\mex$ with $\zero$.
So for example, recalling that $\etime 6 {\mex\mex}= 3+\two$, here is
the corresponding type derivation:

{\footnotesize
\[
\infer{ \ovdash {3+\two} \mex\mex:  \zero }
	{\infer{\Sigma_3\dem\ovdash{2} \lam x.  xx\oplus I: [\two \for\Bt_2, \frac{1}{4}\for \Bt_1]\arrow 
			\zero}{
			\infer[\oplus]{x:[\two \for\Bt_2, \frac{1}{4}\for \Bt_1] \ovdash{1} xx\oplus I:  
				\zero }{
				\infer{x:[ \Bt_2, \two \for\Bt_1] \ovdash  0 xx:    \zero  }{
					x:[  [\two\for \Bt_1]\arrow    \zero  ] \ovdash 0 x:  [ [\two\for \Bt_1]\arrow    \zero   ]  
					& {x:[\two\for \Bt_1]\ovdash 0 x:[\two \for\Bt_1] }
				}
				&
				\infer{\ovdash  0 I: \zero }{}
			}
		} 
		&
		\infer[!]{\Sigma \dem \ovdash{1+\two} \mex:  [\frac{1}{2}\for \Bt_2, \frac{1}{4}\for\Bt_1]}
		{\Sigma_1 \dem  \ovdash {2} \mex:\Bt_1 & \Sigma_2 \ovdash{2} \mex:\Bt_2}    }
	\]}

\subsubsection{Focusing on the Probability of Termination}\label{sec:focusAST}
In this section, we have shown that our type system induces
characterisations of both \AST and \PAST by \emph{the same} family of
derivations, namely the tight derivations.  Moreover, we proved that we
can restrict the search space to the class of null typings whenever
interested in the expected number of steps, only.  But there is more:
if we are interested in the probability of termination \emph{only}, an orthogonal
simplification is possible---we could drop from the typing all the
information on the scaling factors, as that is only used in
deriving the weight.

\section{On Recursion-Theoretic Optimality}\label{sect:optimality}
The uniform characterisation of both forms of termination
we described in Section~\ref{sec:characterization} is remarkable, because one
single system is capable of providing precisely the kind of information
one needs in either case:
\begin{varitemize}
\item
  The (norm of the) underlying type is a lower (but
  tight) bound to the \emph{probability} of termination.
\item
  The weight of type derivations is a lower (but
  again tight) bound to the expected \emph{time} to termination.
\end{varitemize}
As usual in type systems, reasoning is compositional: the typings one
attributes to composite terms are derived from those one assigns to
the subterms. This being said, \AST and \PAST can only be verified
\emph{at the limit}, since all possible type derivations for the given
term and having conclusions of a certain form need, in general, to be taken
into account.

At this point, one may wonder whether one can do \emph{better} than
Theorem~\ref{th:term_tight} when characterising probabilistic termination.  Is it
that one can get away from approximations, and devise a (possibly more
complicated) type system in which \emph{one} type derivation is by
itself a certificate? In this section, we prove that under mild
assumptions in fact one cannot, i.e. that our characterisation is the
best possible, at least recursion-theoretically.

Our results are based on the well-known ones by Kaminski et
al.~\cite{KKM2019}, which establish that in the realm of probabilistic
Turing machines, almost-sure termination is a $\pPi{0}{2}$-complete
problem, while positive almost-sure termination is
$\pSigma{0}{2}$-complete problem. We give two results in this section:
\begin{varitemize}
\item
  On the one hand, we show that probabilistic Turing machines can be
  faithfully encoded into $\PLambda^\cbv$, witnessing the fact that
  the aforementioned recursion-theoretic limitations also hold for
  $\PLambda^\cbv$.
\item
  On the other hand, we prove \emph{by way of our type system} that
  the class of positively almost-surely terminating terms in
  $\PLambda^\cbv$ is $\pSigma{0}{2}$, which in view of the previous point
  means that our type system is \emph{as simple as possible}, recursion-
  theoretically. A similar result is given for
  almost-surely terminating terms and $\pPi{0}{2}$.
\end{varitemize}

\subsection{Probabilistic Turing Machines}

Probabilistic Turing machines~\cite{Santos69,Gill77} (PTMs in the following) can be
defined similarly to ordinary deterministic ones, the main difference
being the fact that the transition function $\delta$ returns not
\emph{one} pair in
$\Sigma\times\{\leftarrow,\downarrow,\rightarrow\}$, but a
\emph{distribution} of those. Various restrictions might be imposed on the
form of those distributions, without affecting the class of
representable (random) functions, but only inducing some
overhead. Here, we assume that the underlying distribution is a
Bernoulli one, assigning probability $\frac{1}{2}$ to one pair and probability
$\frac{1}{2}$ to another one. As usual, we can
also assume to work with $1$-tape Turing machines. Again, this is not
restrictive.  Both notions of termination we have introduced in
Section~\ref{sec:ASTPAST} in the realm of $\PLambdaCBV$ make perfect sense for
Turing machines too, e.g., given a probabilistic Turing machine
$\mathcal{M}$ and an input $x\in\Sigma^*$, we say that $\mathcal{M}$
\emph{is AST on} $x$ if $\mathcal{M}$ converges with probability
$1$. Like ordinary Turing machines, PTMs can be effectively enumerated
and the PTM corresponding to $\alpha$ is indicated as $\mathcal{M}_\alpha$.
This allows us to introduce the following classes of (pairs of) natural
numbers:
\begin{align*}
  \AST_{\mathit{TM}}&=\{(\alpha,x)\mid\mbox{the PTM $\mathcal{M}_\alpha$ is AST on input $x$}\}\\
  \PAST_{\mathit{TM}}&=\{(\alpha,x)\mid\mbox{the PTM $\mathcal{M}_\alpha$ is PAST on input $x$}\}
\end{align*}

\subsection{Encoding PTMs into $\PLambda$}\label{sec:encoding}

Let us now switch to the encoding of probabilistic Turing machines
into $\PLambdaCBV$. As a target language, we actually take a sub-class of terms in
$\PLambdaCBV$, namely the one defined by the following grammar:
\begin{align*}
V &::= x \mid \lambda x.M && \mbox{\textbf{Values}, }\Val^\cps\\
M &::= V \mid MV \mid  M \oplus M && \mbox{\textbf{Terms}, }\PLambda^\cps
\end{align*}
where $MV$ is nothing more than syntactic sugar for $\lett M{xV}$. In
doing so, we follow~\cite{EncodingTuring}, and take as a target calculus for our
encoding one in which only \emph{one} redex is active in any term.
This way, all our results will also be valid in Section~\ref{sec:CBN}, where
intersection types will be generalised to a calculus with call-by-\emph{name}
evaluation.

\newcommand{\enc}[1]{\overline{#1}}
The main ingredients of the encoding are the following ones:
\begin{varitemize}
\item
  States and strings can be encoded following the so-called \emph{Scott
    scheme} \cite{Wadsworth80}, e.g., given an alphabet $\Sigma=\{a_1,\ldots,a_m\}$ strings in $\Sigma^*$
  are encoded following the recursive definition below:
  $$
  \enc{\varepsilon}=\lambda x_1.\cdots.\lambda x_{m}.\lambda y.y
  \qquad\qquad
  \enc{a_i\cdot s}=\lambda x_1.\cdots.\lambda x_m.\lambda y.x_i\enc{s}
  $$
\item
  Similarly, one can encode any tuple of values $(V_1,\ldots,V_n)$
  as $\lambda x.xV_1\cdots V_n$. This encoding easily supports
  projections.
\item
  We can build a fixed-point combinator $Z$ as $MM$, where $M$ is the
  term $\lambda x.\lambda y.y(\lambda z.xxyz)$.  Observe that for
  every value $V$, it holds that $ZV$ deterministically rewrites (in a
  constant amount of steps) to $V(\lambda x.ZVx)$. Notice that the
  argument to $V$ is \emph{not} $ZV$, but is ``wrapped'' into a value
  by way of $\eta$-expansion: this is necessary, given the nature of
  our calculus.
\end{varitemize}
Given the above, and after a fair amount of intermediate technical
results (but closely following~\cite{EncodingTuring}, except in the
encoding of the transition function), one can reach the following:
\begin{theorem}\label{thm:encoding}
  For every probabilistic Turing Machine $\mathcal{M}$, there is
  lambda term $T_{\mathcal{M}}$ such that the evaluation of
  $T_{\mathcal{M}}\enc{s}$ and the computation of $\mathcal{M}$ on
  input $s$ produce the same distributions (up to encodings).
  Moreover, the number of steps taken by $T_{\mathcal{M}}$ is linearly
  related to $\mathcal{M}$. Finally, the term $T_{\mathcal{M}}$ can be
  effectively obtained from (the code of) $\mathcal{M}$.
\end{theorem}

\subsection{Preliminaries from Recursion Theory}

In this subsection, we give some basic definitions about the
arithmetic hierarchy, for the sake of making this paper
self-contained. An excellent reference about these topics
is~\cite{Odifreddi}.

A set $X\subseteq\Nat$ is said to be $\pSigma{0}{n}$ iff
there is a primitive recursive relation $R\subseteq\Nat^{n+1}$
such that
$$
x\in X\Leftrightarrow \underbrace{\exists y_1.\forall y_2\exists y_3.\forall y_4\ldots }_{\mbox{$n$ times}}R(x,y_1,\ldots,y_n)
$$
Dually, $X$ is said to be $\pPi{0}{n}$ iff there is a primitive
recursive relation $R\subseteq\Nat^{n+1}$ such that
$$
x\in X\Leftrightarrow \underbrace{\forall y_1.\exists y_2\forall y_3.\exists y_4\ldots }_{\mbox{$n$ times}}R(x,y_1,\ldots,y_n)
$$
For both the classes $\pSigma{0}{n}$ and $\pPi{0}{n}$, there are related notions of
\emph{hardness}: a set $X\subseteq\Nat$ is $\pSigma{0}{n}$-difficult
(respectively, $\pPi{0}{n}$-difficult) iff it is \emph{at least as
  difficult as} any other $\pSigma{0}{n}$ (respectively, $\pPi{0}{n}$)
problem, i.e. if for every other $\pPi{0}{n}$ problem $Y$ there is a
(recursive) reduction from $Y$ to $X$. Both in $\pSigma{0}{n}$ and in $\pPi{0}{n}$, \emph{completeness} stands for 
containment \emph{and} hardness. These classes form an hierarchy
which is strict; moreover, $\pSigma{0}{n}$ and $\pPi{0}{n}$,
although having non-empty intersections, are incomparable as classes.

Where, in the arithmetical hierarchy, do $\AST_{\mathit{TM}}$
and $\AST_{\mathit{TM}}$ reside? A precise answer to this
question has been given by Kaminski et al.~\cite{KKM2019} in the realm of
while programs, but can easily be rephrased for PTMs:
\begin{theorem}[Kaminski et al.~\cite{KKM2019}]\label{thm:kaminski}
  $\AST_{\mathit{TM}}$ is $\pPi{0}{2}$-complete, while
  $\PAST_{\mathit{TM}}$ is $\pSigma{0}{2}$-complete.
\end{theorem}
Theorem~\ref{thm:kaminski} is quite surprising, in particular
if seen through the lenses of ordinary, deterministic computation.
In universal deterministic computational models (like TMs or the $\lambda$-calculus)
terminating computations form a $\pSigma{0}{1}$-complete set: even
if undecidable, the set is recursively enumerable, and any terminating
computation can be endowed with a finite certificate, itself
(effectively) checkable for correctness. This, by the way, is
a recursive-theoretical justification of the possibility of
building complete systems of intersection types for the deterministic
$\lambda$-calculus in which type derivations play the role of
certificates, as the ones we describe in Section~\ref{sec:gentleintro}: this is possible \emph{only because}
termination is in $\pSigma{0}{1}$.

\subsection{The Optimality Result}

In the probabilistic $\lambda$-calculus, neither form of termination is 
$\pSigma{0}{1}$, and as a consequence type derivations cannot play the role
of certificates. In this section we will formally prove the statement
above, along the lines showing that the form of approximation we
employ is optimal.

First of all, we can give the $\lambda$-counterparts of $\AST_{\mathit{TM}}$
and $\PAST_{\mathit{TM}}$:
$$
\AST_{\lambda}=\{M\mid\mbox{$M$ is AST}\}\qquad\PAST_{\lambda}=\{M\mid\mbox{$M$ is PAST}\}
$$
Theorem~\ref{thm:encoding} and Theorem~\ref{thm:kaminski} together imply that
$\AST_{\lambda}$ is $\pPi{0}{2}$-hard and $\PAST_{\lambda}$
is $\pSigma{0}{2}$-hard. But how about containment?

Actually, our characterisation results , namely Theorem~\ref{thm:characterization} and Corollary~\ref{cor:char_tight} can be
seen as a way to prove that $\AST_{\lambda}$ is \emph{in} $\pPi{0}{2}$
and that $\PAST_{\lambda}$ is \emph{in} $\pSigma{0}{2}$. Indeed, consider
the following two sets
\begin{align*}
  \AST_{\lambda,\vdash}&=\{M\mid\forall r\in\Qnum_{[0,1)}.\exists \Pi.(\Pi\dem\vdash M:\at)\wedge (\norm \at>r)\};\\
  \PAST_{\lambda,\vdash}&=\{M\mid\exists r\in\Qnum.\forall \Pi.(\Pi\dem\ovdash{w}M:\at)\Rightarrow (w<r)\}.
\end{align*}
By Corollary~\ref{cor:char_tight}, $\AST_{\lambda,\vdash}=\AST_{\lambda}$ and
$\PAST_{\lambda,\vdash}=\PAST_{\lambda}$. But \emph{by definition},
$\AST_{\lambda,\vdash}$ is $\pPi{0}{2}$, because checking whether a natural
number is the encoding of a type derivation $\Pi$ having the property that
$(\Pi\dem\vdash M:\at)\wedge \norm \at>r$ for given $M$ and $r$ is
certainly a primitive recursive problem. Similarly for $\PAST_{\lambda,\vdash}$
and $\pSigma{0}{2}$.

This is why we claim that our intersection types are \emph{optimal}: there
cannot be simpler (in the sense of the arithmetical hierarchy)
characterisations of $\AST_\lambda$ and
$\PAST_{\lambda}$.

\renewcommand{\ss}{\ee}

\section{ Variations on the Theme}
This section is devoted to analysing two variations on the type system
we introduced in Section~\ref{sec:intersectiontypes}, itself proved to
satisfy some nice properties, but certainly not being \emph{the only} system
of intersection types one can define in a discrete probabilistic
setting.

\subsection{On Call-by-Name Evaluation}\label{sec:CBN}
Despite the fact that the call-by-value discipline is more natural in
presence of effects, it is legitimate to ask whether the
system of intersection types we have designed can be adapted to CbN
evaluation. This section is devoted to showing that this is
actually the case.

As a language we use here the standard probabilistic untyped 
$\lam$-calculus equipped with weak head reduction, itself already
studied in many papers from the
literature~\cite{DalLagoZorzi,DLSA14}. We first define the
language, called $\PLambda^\cbn$, and its operational semantics, then
the typing system.

\paragraph{The Language of Terms.}
\emph{Terms} and \emph{values} are defined by the grammar
\begin{align*}
  V&::=  x \mid \lambda x.M &&\mbox{\textbf{Values}}, \Val_\oplus^\cbn\\
  M&::=  V\mid MM \mid M \oplus M &&\mbox{\textbf{Terms}}, \PLambda^\cbn
\end{align*}
where $x$ ranges over a countable set of \emph{variables}.
Observe how values are defined as in $\PLambda^\cbv$, while
terms are slightly different, and more in line with the
usual $\lambda$-calculus. Another remark: the class $\PLambda^\cps$
is trivially a subclass of $\PLambda^\cbn$.

\paragraph{Operational Semantics and Probabilistic Termination}
As in CbV, we first define a one-step reduction relation $\red$ from terms to
multidistributions. The rules are given in
Figure~\ref{fig:reductions}.  We then lift $\red$ to a \emph{reduction
  of multidistributions}, and this can be done as for $\PLambda^\cbv$, so following
the rules in Figure \ref{fig:lifting}.
\begin{figure}\centering
    \begin{minipage}{0.97\textwidth}
      \[
      \infer[\beta]{(\lam x.M)V\red \mdist{M\subs x V}}{}  \qquad\qquad\qquad
      \infer[\mathtt{head}]{N M  \red \mdist{p_i (N_i M)}_{\iI}}{N\red \mdist{p_iN_i}_{\iI}}
      \]
  \end{minipage}
  \caption{Reduction Steps }\label{fig:reductions}
\end{figure}
Values are precisely the closed terms which cannot be further
reduced. The definitions of $ \nnorm {\m_k}$, $\eval k M$, $\ProbTerm
M$, $\etime k M$, and $\ETime M$ can be given exactly as in Section~\ref{sec:ASTPAST}  
Again, observe how the semantics of all terms of $\PLambda^\cps$ is
the same if defined through CbV, as we did originally, or through
CbN, as we are doing here. As a consequence, all results from
Section~\ref{sec:encoding} also hold for CbN.
\subsubsection{The Type System}
Non-Idempotent Intersection types for the Call-by-Name $\lam$-calculus
\cite{Gardner94, Kfoury00, MNM2004, Carvalho2018} are well-studied.
We adapt them to our probabilistic setting. The types reflect the
underlying dynamics, which is simpler than that of CbV, since a term
cannot be copied once evaluated. Like in the case of $\PLambda^\cbv$,
the type system is based on \emph{three}, rather than \emph{two}
layers, namely arrows, intersection types, and multidistribution
types.  Notice however that now a \emph{type distribution} is a
(multi)-distribution over arrows. An \emph{intersection type} is a
multiset of scaled types, \ie a multiset of pairs $q\for \at$ where
$\at$ is a type distribution, and $q\in (0,1]\cap \Qnum$ is as usual a
  \emph{scale factor}. Types are defined by means of the following
  grammar:
\begin{align*}
\At,\Bt & ::=  * \mid {\ma \arrow \at}  && {\mbox{ \textbf{Arrow Types}}}\\
\mA,\mB & ::=  \mul{q_1\for \at_{1},...,q_n \for\at_{n}} n \geq 0 && \mbox{  \textbf{Intersection Types}}\\
\at,\bt & ::=  \mdist{p_{1}\At_{1},...,p_{n}\At_{n}}, n \geq 0 && \mbox{ \textbf{Type Distributions}}
\end{align*}
Observe the presence of the special arrow type $*$, which here plays the
role of the empty multiset $\mul{}$ in CbV.
\paragraph{Typing Rules.}
The type assignment system in Figure \ref{fig:WWcbv} proves judgments
of the shape $\Gamma \oder w M: \type$, where $\Gamma$ is a type
context, $M$ a term, $w\in \Qnum $ is a  counter, and $\type$ is either $ \at$ or $\mA$. 
The notation  $q\for\Gamma$ is as in Section~\ref{sec:typesCBV}, taking into account that now 
if  $\mA=\mset{q_i\for\at_i}_{\iI}$,  $u\for \mA$  is   $ \mset{(u q_i)\for\at_i}_{\iI}$.
\begin{figure}
  \centering
		\begin{minipage}{0.95\textwidth}
		\vspace{8pt}
				 	\[   \infer[\TVar]{ x:\mset {1\for\at}  \ovdash 0 x:\at}{}  
				\quad\quad\quad
				\infer[\TVal]{ \ovdash 0  \lam x.M :\dist{*}}{}
				\quad \quad\quad \infer[\TZero]{  \oder{0 } M:\zero}{}
				\]

				\[	\infer[\lambda]{\Gamma \ovdash {w+1} \lam x. M: \dist{\ma \arrow \bt}}{\Gamma, x:\ma \ovdash {w} M:\bt} \quad\quad
				\infer[@]{\Gamma  \uplus_k  {p_k\for}\Delta_k \oder {w+ \sum_k p_k w_k} MN: {\dsum_k~ p_k \bt_k}}
				{\Gamma \oder {w} M: \dist {p_k (\ma_k \arrow \bt_k) }_{\kK}&  
					\big(\Pi_k\dem ~\Delta_k \oder {w_k} N:  \mA_k\big)_{k\in K}}	
				\]

				\[
				\infer[!]{\uplus_i (q_i\for\Gamma_i)\oder {\sum_i  q_i w_i} M:\mset{q_i\for\at_i}_{\iI}}
				{(\Gamma_i\oder {w_i} M: \at_i)_{\iI}\quad &  \quad(q_i)_{\iI} \mbox{ scale factors}
				}
			\quad
				\infer[\oplus]{\two\for\Gamma \uplus \two\for\Delta \oder{1+\two w_1 +\two w_2} M\oplus N: \two \at \dsum \two \bt}
				{\Gamma \oder {w_1} M: \at & \Delta \oder{w_2} N: \bt} 
				\]
\vspace{1pt}				
				
		\end{minipage}
       	\caption{Non-Idempotent Intersection Type Rules for $\PLambda^\cbn$}\label{fig:WWcbn}
                                
\end{figure}
The notion of a tight type needs to be appropriately adapted.
\begin{Def}[Tight Types and  Derivations]
  A \emph{type} $\at$ is said to be \emph{tight} if it is a
  multidistribution on the arrow type $*$. A derivation $\Pi\dem ~
  \oder w M: \at $ is tight whenever $\at$ is tight.
\end{Def}

\paragraph{Basic Properties.}
As in CbV, some basic properties of the type system are not only
useful, but reveal the nature of the type system. First of all, any
closed value $V$ can be tightly typed with probability $1$, by
$ \oder 0 V :\dist{*} $. Moreover, a degenerate form of the
rule $!$ allows us to derive the following  \emph{for any term} $M$:
\[\infer{\oder{0} M:[]}{}\]
Finally, a useful instance of the $@$ rule is the following:
\[\infer{\Gamma \oder {w} MN:\zero}{\Gamma \oder{w} M:\zero }\]
	
\subsubsection{Characterising CbN Probabilistic Termination}
The just introduced type system allows us to transfer all results from Section~\ref{sec:characterization}
to $\PLambda^\cbn$.  Finitary soundness and finitary completeness
both hold, exactly as in Theorem~\ref{th:correct} and
Theorem~\ref{th:complete}. The statement is the same, taking into
account that now $M$ is a closed term of $\PLambda^\cbn$. As a
consequence, we can:
\begin{varitemize}
\item
  on the one hand characterise \AST and \PAST in a uniform way, via
  \emph{tight typing}, exactly as in Theorem~\ref{th:term_tight} and
  Corollary~\ref{cor:char_tight};
\item
  on the other hand characterise \PAST via \emph{null typing},
  this time exactly like in Theorem~\ref{thm:char_null}.
\end{varitemize}
As mentioned in Section~\ref{sec:focusAST}, one can also obtain a (simpler) type
system for \AST by dropping from the typing all the information on the
scaling factors.
\subsection{Multidistributions vs. Distributions}\label{sec:multi}
In the design of any type system, several choices are possible. Some
are a matter of taste, some other are crucial. In this section, we
discuss a choice we have implicitly made throughout the paper, namely
the use of multidistributions in types. One may legitimately
wonder if we could use \emph{distributions} of types instead of
multidistributions.  Actually, it turns out that multidistributions
are necessary to obtain a perfect match between typing and
termination. This choice is in fact crucial for \emph{completeness} to
hold in the call-by-value typing system. Let us see why.

Consider a term in the form $\lett N M$. Since the argument $N$ is
typed with a multidistribution $\ct =\mdist{p_k\mA_k}_{\kK}$, the
continuation $M$ must be able to receive any $\mA_k$. Indeed, the
typing rule $\letr$ asks for type derivations having conclusion
$x:\mA_k\der M:\bt_k$ for each $k$. Each value of $k$ indeed
corresponds to one of the possible probabilistic evolutions of $N$,
due to the use of multidistributions, in which collapsing two elements
of $\mA_h$ and $\mA_l$ in $\ct$ is simply not possible. Going to
distributions, thus allowing for such a collapse, would not be a
problem for \emph{soundness}, but we would loose the properties of
weighted subject expansion (Lemma~\ref{lem:SubEx}) on which
\emph{completeness} relies. We now see why by way of a concrete
example.
\begin{example}[Weighted Subject Expansion relies on multidistributions]
  Assume $ P \red \mdist{\two P_1, \two P_2}$. The claim of weighted
  subject expansion is that, given derivations $\Pi_i \dem \oder {w_i}
  P_i:\bt_i $ for each $i$, we can obtain a derivation $\Pi\dem \oder
  {w_i} P:\bt $, where $w\geq 1+\sum \two w_i$ and $\bt= \two \bt_1 \dsum \two
  \bt_2$. Weighted subject expansion is
  proved by  induction on the structure of the reduction $\red$.
  The key point is the $\letr C$ rule. Let us focus on it.  Consider
  $P\deff (\lett {N_1\oplus N_2} M)$ and so $P_i\deff (\lett { N_i}
  M)$, and consider the following type derivations
  for $P_1$ and $P_2$.
  \[
  \infer[\texttt{let}]{
    \oder {w_1 } \lett  {N_1} M: \bt_1}
        {\oder {v_1} N_1:  \dist {  \mA } & \Pi_1\dem x:  \mA\oder {u_1} M: \bt_1}	
        \quad\quad
  \infer[\texttt{let}]{
	  \oder {w_2 } \lett  {N_2} M:\two \bt_2}
              {\oder {v_2} N_2:  \dist { \two \mA} & \Pi_2\dem x:  \mA\oder {u_2} M: \bt_2}	
    \]	
  By definition, $ P \red \mdist{\two P_1, \two P_2}$ is derived as
  follow:
  \[
  \infer[\letr C]{(\lett {N_1\oplus N_2} M)  \red \mdist{\two (\lett {N_1} M), \two (\lett {N_2} M)}   }
        {N\red \mdist{\two N_1, \two N_2}}
  \]
  and we would like to derive a type derivation for $P$ out of all this. By \ih, since
  $N\red \mdist{\two N_i}_{\iI}$, we can assume that there exists a
  derivation $\Phi\dem \oder{v} N:\ct $ such that $v\geq 1+\sum \two v_i$
  and $\ct= \two \mdist{\mA} \dsum \two \mdist{\two \mA} = \mdist{\two
    \mA, \four \mA}$.
  And indeed, by collecting $\Pi_1$ and $\Pi_2$,
  we have a derivation which satisfies the claim
  \[
  \infer[\texttt{let}]{\oder {w} \lett  {N} M:\two \bt_1 \dsum \four \bt_2 }
        {\oder {v} N:    \mdist{\two \mA, \four \mA} & x: (\mA\oder {u_i} M: \bt_i)_{i\in \{1,2\}}}	
  \]	
  This is possible precisely because---due to the adoption of
  \emph{multidistributions}---the two occurrences of $\mA$ are kept
  separated:
   notice that $\bt_1$ and $\bt_2$
  may be very different types.  If we worked with \emph{distributions},
  this information would be irremediably lost.
  By \ih, we would have \emph{just one} derivation $\Phi\dem \oder{v}
  N:\ct $ where $\ct= \two \mdist{\mA} + \two \mdist{\two \mA} =
  \mdist{\frac{3}{4} \mA} $. We would like to build
  the following derivation:
    \[
    \infer[\texttt{let}]{\oder {w} \lett  {N} M:\two \bt_1 \dsum \four \bt_2 }
    {  
    	\oder {v} N: \mdist{\frac{3}{4} \mA} \quad & \quad \Pi}	
    \]    
  How could we build $\Pi$, however? There is no way to \emph{merge}
  the two derivations $\Pi_1$ and $\Pi_2$, so the type system would
  need to be substantially reengineered.
\end{example}
This issue only affects call-by-value evaluation, which is more complex than
call-by-name, but also more expressive in a setting with effects.
In CbN, choosing distributions would not impact the
results, because evaluating a term \emph{before} copying it
(i.e. before using it in possibly many different ways)
is simply impossible.

\section{Related Work}
Systems of types for probabilistic programs exist in the
literature. In particular, sized types~\cite{HPS1996}, and linear
dependent types~\cite{DLG2011} have been generalised to probabilistic
programming languages, and have been proved to be sound methodologies
for checking almost-sure termination~\cite{DLG2019} and positive
almost-sure termination~\cite{ADLG2019} in an higher-order
setting. None of such systems is complete, however. Recently, Breuvart
and Dal Lago~\cite{BDL2018} introduced systems of intersection types
which are sound and complete as a way of deriving the probability of
convergence of terms in probabilistic lambda-calculi. However, the
number of reduction steps to normal form is not kept track of by
types, due to the nature of the intersection operator, which in Dal
Lago and Breuvart's system is idempotent. Moreover, relying on distributions
(instead of multidistributions) of types makes call-by-value evaluation
harder to deal with, and ultimately results in a rather convoluted set of typing
rules.

Intersection types have been pioneered by Coppo and
Dezani~\cite{CoppoDezani1978,CoppoDezani1980}, and developed in a
series of papers in which various notions of termination for the
$\lambda$-calculus were characterised, and the
relationship with denotational semantics was thoroughly
investigated~\cite{CDV1980, Pottinger1980, BCDC1983, CDCZ1987}.  They have also been
extended to calculi besides the $\lambda$-calculus, like
$\lambda\mu$-calculi~\cite{vBBL13} or object
calculi~\cite{deLiguoro2001}. Besides the already discussed work by
Breuvart and Dal Lago~\cite{BDL2018}, one should also mention the
work by de' Liguoro and
colleagues~\cite{DCLP1993,deLiguoroPiperno1995} about filter models
and intersection type assignment systems for extensions of
$\lambda$-calculi with nondeterministic choice operators, whose
semantics is however fundamentally different than that of the
probabilistic choice operator we consider here: in the former
one observes \emph{may} or \emph{must} convergence (or combinations
thereof), while here the notion of observation is genuinely
quantitative.

Non-idempotent intersection types have been known since the work by Gardner~\cite{Gardner94},
studied in connection with expansion variables by Carlier et al.~\cite{CPWK2004},
and further analysed in their relation to normalisation by Mairson and
M\"oller-Neergard~\cite{MNM2004}. The precise correspondence between non-idempotent
intersection type system derivations and the number of reduction steps
necessary to normalise the underlying term has been first noticed
by De Carvalho~\cite{Carvalho2018}, and further refined by Bernadet and
Lengrand~\cite{BL2013}, and later by Accattoli et al.~\cite{AGK2018}, and~\cite{AGL19}, the
latter being a source of inspiration for this work in its
reflecting weak notions of reduction inside intersection types.
All these contributions, however, deal with deterministic $\lambda$-calculi.

Formal verification techniques for probabilistic termination and
complexity analysis are plentiful, and ranges from model
checking~\cite{EY2009,KDLG2019} to abstract
interpretation~\cite{Monniaux2001}, to the ranking
supermartingales~\cite{ChakarovS13}, to amortised
analysis~\cite{NCH2018} to the interpretation method from term
rewriting~\cite{ADLY2020}. The only methodology among these that, at
least so far, has been employed for the analysis of higher-order
probabilistic programs is the one by Kobayashi et al.~\cite{KDLG2019},
which deals with probabilistic variations on higher-order recursion
schemes. Some of the ideas which we introduced in the paper are indeed
variations of similar ones from the imperative setting (e.g. the
handling of expectations by way of a quantity which decreases on the
average). The presence of higher-order functions, however, forced us
to develop new tools, since types must be more informative than just,
say, ranking supermartingales. Not only the \emph{value} or the
\emph{size} of the input matter, but also how the input \emph{behaves}
turns out to be crucial, given that it can potentially be used as a
function. Looking at all this from a different perspective, we can
safely say that higher-order probabilistic programs could of course be
verified by translating the input program into a first-order
equivalent one, then applying state-of-the art techniques designed for
such a setting (e.g.,~\cite{MM05,KKMO18}). The main advantage of
thinking in terms of types, however, is that the underlying
verification problem can be tackled \emph{compositionally}, so
allowing for a modular analysis. In presence of higher-order
functions, one has to prove something stronger than the mere
underlying termination property, namely that the program at hand
satisfies the property when seen in isolation, but also behaves well
when fed with functional inputs, provided those functions behave well
themselves. Verification techniques designed for first-order programs
are not designed with all this in mind, and encoded higher-order
programs would thus be harder to verify.

The operational and denotational semantics of probabilistic $\lambda$-calculi
have been studied thoroughly themselves, starting from the pioneering contributions by
Sahed-Djaromi~\cite{SahebDjahromi1978} and Jones and Plotkin~\cite{JonesPlotkin1989}.
Noticeably, Ehrhard et al.'s probabilistic coherent spaces~\cite{ETP14} can be presented
as a non-idempotent intersection type system which, being inherently
semantic, is fundamentally different from the one we have here: no result
is given about the expected time to termination of the interpreted terms,
and results like those we proved in Section~\ref{sect:optimality} would
be much harder to get.


\section{Conclusion}

This paper introduces and studies non-idempotent intersection type
assignment systems for probabilistic $\lambda$-calculi, showing they
can precisely characterise the expected runtime \emph{and} the
probability of termination within a single framework, despite them
having incomparable recursion-theoretic difficulties, and thus an
inherently different nature. The key ingredients are non-idempotency
and scaling. Noticeably, the same ideas work in the call-by-name and
call-by-value paradigms.

The system of intersection types we have introduced in this work
should be conceived as a tool for the theoretical analysis of a
phenomenon, rather than as a proper verification technique: type
inference is for obvious reasons highly undecidable. This does not
mean, however, that the same necessarily holds in restricted calculi,
as witnessed by the fruitful use of intersection types as a
verification tool in subrecursive deterministic
lambda-calculi~\cite{KobayashiOng,Kobayashi,KfouryWells1999}.  As a
consequence, it would be very interesting, e.g., to study which
fragments of $\PLambda^\cbv$ and $\PLambda^\cbn$ are expressive enough
to capture recursive Markov chains~\cite{EY2009}, in which almost-sure
termination is known to be decidable

The absence of idempotency---an essential ingredient indeed---can be
seen in two different forms, namely in intersection types, where union
is not an idempotent operation, and in distribution types, which are
taken as \emph{multi}distributions and which thus \emph{do not} form a
barycentric algebra, precisely due to the failure of idempotency. A
thorough study of this phenomenon, together with an analysis of the
relationship between this work and the denotational semantics of
probabilistic $\lambda$-calculi is outside the scope of this paper,
but it is certainly something the authors would like to pursue in the
foreseeable future.

\begin{acks}
  This work was partially supported by
  \grantsponsor{ANR}{ANR}{https://anr.fr/} PRC project PPS
  (\grantnum{ANR}{ANR-19-CE48-0014}), by
  \grantsponsor{ERC}{ERC}{http://erc.europa.eu} Consolidator Grant
  DIAPASoN (\grantnum{ERC}{818616}), and by
  \grantsponsor{MIUR}{MIUR}{http://miur.gov.it} PRIN ASPRA
  (\grantnum{MIUR}{201784YSZ5}).
\end{acks}

\bibliography{biblio}

	\newpage
	\appendix

\section*{APPENDIX}

\section{Proofs of Subject Reduction and Subject Expansion}

A type ( $ \type ::= \At\mid \mA \mid\at$ ) is as defined in Section~\ref{sec:typesCBV}.

\subsection{Proof of Weighted Subject Reduction  (Lemma~\ref{lem:SubRed})}
As usual, the proof of subject reduction relies on a substitution lemma.
\begin{lemma}[Substitution Lemma]\label{lem:sub}
	If there exist derivations  $\Pi \dem \Gamma, ~z:\mC \oder{w} M:\type$, and $\Phi \dem  \oder{v} V:\mC$, 
	then exists a derivation  $\Pi' \dem \Gamma   \oder{w'} M\subs z V : \type$. 
	Moreover 
	  \begin{enumerate}
		\item $w'=w+v$ and 
\item  $\sz{\Pi'} \leq  \sz{\Pi}+ \sz{\Phi}$. 
	\end{enumerate}
	
\end{lemma}
\begin{proof}
	The proof is by induction on the derivation  $\Pi$; we examine the last rule.
We write only  the  key  cases.	

\begin{itemize}

	\item Rule $\letr$. Assume $\Pi$ is as follows
\[
\infer[\texttt{let}]{z:\mC_0\uplus_k p_k\for \mC_k, \Gamma  \uplus_k  {p_k\for}\Delta_k \oder {w=w_0+ \sum_k {p_k w_k}
		+ {1}} \lett  N M: {\dsum_k~ p_k \bt_k}}
{ z:\mC_0,\Gamma \oder {w_0} N: \dist { p_k \mA_k }_{\kK}   &	(z:\mC_k, \Delta_k,  x:  \mA_k\oder {w_k} M: \bt_k)_{k\in K}	   	}	
\]	
and $\Phi \dem  \oder{v} V:\mC= \mC_0\uplus_k p_k\for \mC_k$.
By \emph{partition and Scaling (Lemma~\ref{lem:scale})}, we have $v=v_0 +\sum_k v_k$ and 
\[   \oder {v_0} V: \mC_0  \quad \Big( \oder{\frac {v_k}{p_k}} V:   \mC_k \Big)  \]

By \ih, we obtain the following derivation
\[
\infer[\texttt{let}]{ \Gamma  \uplus_k  {p_k\for}\Delta_k \oder {w'} (\lett  N M)\subs{z}{V}: {\dsum_k~ p_k \bt_k}}
{ \Gamma \oder {w_0+v_0} N\subs{z}{V}: \dist { p_k \mA_k }_{\kK}   &	(\Delta_k,  x:  \mA_k\oder {w_k +\frac{v_k}{p_k}} M\subs z V: \bt_k)_{k\in K}	   
		}	
\]	
where $ w'=(w_0+ v_0) + \sum_k p_k( w_k + \frac{v_k}{p_k})  +1 $.

We verify that claims 1. and 2. hold.
\begin{enumerate}
	\item $w'= w_0 +  \sum_k p_k( w_k ) +1+  v_0 + \sum_k p_k(\frac{v_k}{p_k}) = w+v$
	\item immediate.
\end{enumerate}

\item Rules $\bm{!}$ and   $\oplus$ also use partition and Scaling (Lemma~\ref{lem:scale}).

\item Rules $\lam$ and $@$ are as usual.
\end{itemize}
\end{proof}

\begin{lemma*}[\ref{lem:SubRed}. Weighted Subject Reduction]
	Suppose that $\Pi\dem \oder{w} P:\bt$, with $w >0$, and that $P \red
	\mdist{q_iP_i}_{\iI}$. Then for every $\iI$ there exists a
	derivation $\Pi_i$ such that $\Pi_{i}\dem\oder {w_i} P_i: \bt_{i}$,
	and $\size{\Pi} > \size{\Pi_i}$. Moreover:
	\begin{enumerate}
		\item
		$\bt= \dsum_{i\in I} q_i \bt_{i }$ 
		\item
		$w=1+ \sum_{i\in I} q_i w_i$		
	\end{enumerate}
\end{lemma*}


\begin{proof}Observe that $w\not=0$ implies that the last rule of the derivation  $\Pi$ is not a $\TZero$-rule.
	The proof is by induction on the definition   of the reduction step  $P \red \mdist{q_iP_i}_{\iI}$.

	\begin{itemize}
	\item 	 $\bm{\beta}$.  Assume $P:= (\lam x.M)V $ and 
	\[
	\infer[\beta]{(\lam x.M)V\red \mdist{M\subs x V}}{}  
	\]
	We examine the derivation $\Pi$, and conclude by Substitution Lemma.

	\item 	$\bm{\letr V}$.  Assume $P\deff (\lett V M)$ and $P\red \mdist {1P_0}$ as follows
	\[\infer[\letr V]{\lett V M  \red   \mdist{M \subs x V}}{}\]
	By assumption,  there exists $\Pi\dem \oder {w}P:\at$, which must have the following shape 
	
	\[
	\infer[\texttt{let}]{\oder {w_1+w_2  + 1} \lett  V M: \at}
	{\oder {w_2} V:\mdist{\mC} &	  x:  \mC\oder {w_1} M: \at	    	}	
	\]
	
	 By Substitution Lemma, there exists $\Pi' \dem    \oder{w_1+w_2} M\subs x V : \at$. Observe  that $I=\{0\}$ is a singleton, and 
	 $q_0=1$. Conditions 1.,2.,3. are all satisfied.

	\item 	 $\bm{\oplus}$.  Assume $P:= P_1\oplus P_2$ and $P\red\mdist{\two P_1,\two P_2}$ as follows   
	\[\infer[\oplus]{P_1\oplus P_2 \red \mdist{\two P_1, \two P_2}}{}  \]
	By assumption,  there exists a type derivation  $\Pi\dem \oder {w}P_1\oplus P_2:\at$, which must have the following shape 
	
		\[
	\infer[\oplus]{\oder{1+\two w_1 +\two w_2} P_1\oplus P_2: \at=(\two \at_1 \dsum \two \at_2)}
	{ \oder {w_1} P_1: \at_1 &  \oder{w_2}  P_2: \at_2} 
	\]
	All the points in the claim hold.

	\item 	 $\bm{\letr C}$. 
	Assume $P \deff (\lett N M)$ and  $P\red \mdist{q_iP_i}$ , with $P_i \deff  (\lett {N_i} M)$, according to 
	\[ \infer[\letr C]{(\lett N M)  \red \mdist{q_i (\lett {N_i} M)}}{N\red \mdist{q_iN_i}_{\iI}}. \]
	By assumption,  there exists a type derivation $\Pi\dem \oder {w}P:\bt$, which must have the following shape

	\[
	\infer[\texttt{let}]{  \oder {w=v+ \sum_{k\in K} p_k v_k +  1  } \lett  N M: \bt= ({\dsum_{\kK}~ p_k \bt_k})}
	{\Phi\dem  \oder {v} N: \at=\dist { p_k \mA_k }_{\kK}    &	\big( x:  \mA_k\oder {v_k} M: \bt_k\big)_{k\in K}	   	}	
	\]		

By \ih, there exist type derivations  $\Phi_i \dem \oder {v_i} N_i:\at_i $  such that 	$\size{\Phi} > \size{\Phi_i}, ~\forall i\in I$ and moreover
\begin{enumerate}
 \item $\at= \dsum_{\iI} q_i \at_{i }$;
\item   $v=1+ \sum_{\iI} q_i v_i$;
\end{enumerate}
By  point 1. above,   $\dsum_{\iI}q_i \at_i = \dist { p_k \mA_k }_{\kK}$. As $K$ is the index set of $\at$, let $K_i$ be the index set of $\at_i (\iI)$  (notice that for some $i$, it may be  possible  $K_i=\emptyset$). 
We have  $K=\uplus_{\iI} K_i$ and  $q_i \at_i=\mdist{p_k\mA_k}_{k\in K_i}$; therefore $\at_i=\mdist{\frac{p_k}{q_i}\mA_k}_{k\in K_i}$. 
For each $\iI$ we obtain  the following type  derivation $\Pi_i$:

	\[
\infer[\texttt{let}]{  \oder {w_i=v_i+ \sum_{k\in K_i} \frac{p_k}{q_i}v_k   + 1
		   } \lett  {N_i} M: \bt_i= 
	({\dsum_{k\in K_i}~ \frac{p_k}{q_i} \bt_k})}
{\Phi_i\dem  \oder {v_i} N_i: \at_i=\dist { \frac{p_k}{q_i} \mA_k }_{k\in K_i}    &	\big( x:  \mA_k\oder {v_k} M: \bt_k\big)_{k\in K_i}	   	}	
\]		
We check that point 1. and 2. of the claim are verified.
\begin{enumerate}
	 \item $ \dsum_{\iI}  q_i \bt_{i } = \dsum_{\iI}  q_i  ({\dsum_{k\in K_i}~ \frac{p_k}{q_i} \bt_k}) = 
	 \dsum_{\iI}   ({\dsum_{k\in K_i}~ q_i \frac{p_k}{q_i} \bt_k}) =  \dsum_{k\in K}~ {p_k}\bt_k  = \bt$ 
	\item 
	
{\scriptsize 	\begin{align*}
	1+ \sum_{\iI}  q_i w_i \\
	=  1+  \sum_{\iI}  q_i  v_i +
\sum_{\iI}  q_i \big(   \sum_{k\in K_i} {\frac{p_k}{q_i} v_k}   \big)
+ 	 \sum_{\iI}  q_i\\
= (1+  \sum q_i  v_i )+ \sum_i  \big(   \sum_{k\in K_i} {q_i \frac{p_k}{q_i} v_k}   \big)  
+  \sum_{\iI}  q_i \\
= v + \sum_{k\in K}  p_k v_k  +  { \sum_{\iI}  q_i }\\
= v + \sum_{k\in K}  p_k v_k  +  {1}\\
 = w
	\end{align*}}
%
%
%
where 	$\sum_{\iI}  q_i =1$ is by definition of  $P \red
\mdist{q_iP_i}_{\iI}$.
\end{enumerate}
\end{itemize}
\end{proof}

\subsection{Proof of Weighted Subject Expansion (Lemma~\ref{lem:SubEx})}
The proof of Subject expansion relies on an anti-substitution lemma, whose proof is routine.

\begin{lemma}[Anti-substitution] Assume  $\Phi\dem ~ \Gamma \oder {w}  M\subs{x}{V}:\at$, where $V$ is  closed. Then there exist:
	\begin{enumerate}
		\item an intersection type $\mC$;
		\item a derivation  $\Phi_1 \dem \Gamma, x: \mC \oder{w_1} M:\at$
		\item a derivation $\Phi_2 \dem  \oder{w_2} V:\mC$
	\end{enumerate}
	such that 
	\begin{itemize}
		 \item $\at= \dsum q_i \at_{i }$ 
		\item $w= w_1+w_2$
		
	\end{itemize}
\end{lemma}

%
%
%

\begin{property}\label{lem:zero_app} If  $\Pi \dem \oder {w+1} \lett N M:\zero$, 
	then $\Pi$ has either of the following shapes:
	
{\footnotesize 	\[ \infer[\letr]{\Pi_i \dem \oder {w+1} \lett  N M:\zero}{\oder w N:\zero}
	\quad \quad\quad\quad
\infer[\texttt{let}]{ \Pi_i \dem 
	\oder {w_i=v_i+ \sum_{k\in K_i} p_k v_k + 1} \lett  {N_i} M: \zero}
{ \oder {v_i} N_i: \dist { p_k \mA_k }_{k\in K_i}    &	\big(   x:  \mA_k\oder {v_k} M: \zero \big)_{k\in K_i}	   	}	
\]	}
\end{property}

\begin{lemma*}[\ref{lem:SubEx}. Weighted Subject Expansion]
	Let $P$ be a closed term.  Assume that $P\red \mdist{q_i P_i}_{\iI}$ and
	that for each $\iI$, $\Pi_i \dem \oder {w_i} P_i:\at_i $. Then, there exists a
	single derivation $\Pi\dem \oder {w} P:\at $ such that $\at=\dsum_i
	(q_i\at_i)$ and {$w\geq 1+\sum q_i w_i$}.
\end{lemma*}

\begin{proof}The proof is by induction on the definition  of the reduction step $P\red \mdist{q_i P_i}_{\iI}$.

	\begin{itemize}
\item 	 $\bm{\beta}$.  Let $P:= (\lam x.M)V $ and 
\[
\infer[\beta]{(\lam x.M)V\red \mdist{M\subs x V}}{}  
\]

By assumption,  there exists $\Pi'\dem \oder {w'} M \subs x V:\at$. By anti-substitution Lemma, there exist derivations $\Phi_1,\Phi_2$ which allow for the following inference, where $w'=w_1+w_2$ 

\[
\infer[@]{\oder {1+w_1+w_2}  (\lam x.M)V :\at}
{\infer[!]{\oder{1+w_1} \lam x.M:[\mC\arrow \at]}
	{\infer[\lam ]{ \oder{1+w_1} \lam x.M:\mC\arrow \at }  {\Phi_1 \dem x:  \mC\oder {w_1} M: \at	  } }
	&  \Phi_2 \dem \oder {w_2} V:\mC 	}	
\]	

\item 	 $\bm{\letr V}$.  Let  $P:= (\lett V M)$ and \[\infer[\letr V]{\lett V M  \red   \mdist{M \subs x V}}{}\]
{By assumption,  there exists $\Pi'\dem \oder {w'} M \subs x V:\at$.
	By anti-substitution Lemma, there exist derivations $\Phi_1,\Phi_2$ which allow for the following inference, where $w'=w_1+w_2$ 

\[
\infer[\texttt{let}]{\oder {{w_1+w_2 {+1}}} \lett  V M: \at}
{\infer {\oder {w_2} V:\mdist{\mC}}{\Phi_2\dem\oder {w_2} V:{\mC} }    &    \Phi_1\dem	  x:  \mC\oder {w_1} M: \at	   	}	
\]	}

\item 	 $\bm{\oplus}$. Let $P:= P_1\oplus P_2$ and 
\[\infer[\oplus]{P_1\oplus P_2 \red \mdist{\two P_1, \two P_2}}{}  \]
By assumption,  for each $\iI$, there exists $\Pi_i \dem \oder {w_i} P_i:\at_i $. Therefore, we obtain
	\[
\infer[\oplus]{\oder{1+\two w_1 +\two w_2} P_1\oplus P_2: \two \at_1 \dsum \two \at_2}
{ \oder {w_1} P_1: \at_1 &  \oder{w_2}  P_2: \at_2} 
\]

\item 	 $\bm{\letr C}$.
		Let  $P:= (\lett N M)$, $P_i := (\lett {N_i} M)$,  and 
	   \[ \infer[\letr C]{(\lett N M)  \red \mdist{q_i (\lett {N_i} M)}}{N\red \mdist{q_iN_i}_{\iI}} \].
	
	By assumption,  for each $\iI$, there exists a derivation  $\Pi_i \dem \oder {w_i} P_i:\at_i$.
Each	$\Pi_i$  has one of the following shapes:

\begin{enumerate}
	\item \[\infer[\TZero]{ \oder {0}  (\lett {N_i} M):\zero}{}\]
	\item 	
	\[
	\infer[\texttt{let}]{ \Pi_i \dem 
		 \oder {w_i=v_i+ \sum_{k\in K_i} p_k v_k + 1 } \lett  {N_i} M:(  {\dsum_{k\in K_i}~ p_k \bt_k})}
	{\Phi_i\dem \oder {v_i} N_i:   \dist { p_k \mA_k }_{k\in K_i} :=\ct_i   &	(\Theta_k  \dem x:  \mA_k\oder {v_k} M: \bt_k)_{k\in K_i}	   	}	
	\]	
	Note that  $K_i$ may  be $\emptyset$ for some $i$, and that  $\bt_k$ may be $\zero$ for some $k$.
\end{enumerate}

Let us examine the three sub-cases.
\begin{itemize}

\item[(i)]  Suppose that all  $\Pi_i$ are as  (1.). The following $\Pi$  satisfies the claim 
\[\infer{\oder 1  \lett N M :\zero}{\oder 0  N:\zero}\]

\item[(ii)] Suppose all $\Pi_i$ are as (2.)	
By \ih, since $N\red \mdist{q_iN_i}_{\iI}$,  there exists a derivation $\Phi\dem \oder{v} N:\ct $ such that $v \geq 1+\sum q_i v_i$ and
 	$\ct= 
 	\dsum_{\iI} \big(q_i 	\mdist{p_k \mA_k}_{k\in {K_i}}\big) $,  	
 	that is $\ct=  \dsum_{\iI} 	\mdist{q_i p_k \mA_k}_{k\in {K_i}} $ 

 Therefore, by collecting  all premisses $\Theta_k$ for $k\in \biguplus_{\iI} K_i$, we obtain

	\[
\infer[\texttt{let}]{ \Pi \dem  \oder {w} \lett  {N} M: {\dsum_{\iI}~ q_i (\dsum_{k\in K_i}p_k \bt_k)}}
{   \Phi\dem \oder {v} N:  \dsum_{\iI} 	\mdist{q_i p_k \mA_k}_{k\in {K_i}}   &	(  x:  \mA_k\oder {v_k} M: \bt_k)_{k\in \biguplus_{\iI} K_i}	   	}	
\]	
where $w~=~v+ \sum_i  (\sum_{k\in K_i}{q_ip_k}v_k    ) +  1$. Since $v \geq 1+\sum q_i v_i$
the claim holds because:
{\scriptsize \begin{align*}
1+ \sum_{\iI}  q_i w_i =\\
1+  \sum_{\iI}  q_i  v_i +
\sum_{\iI}  q_i \big(   \sum_{k\in K_i} {p_k v_k}   \big) + 	 \sum_{\iI}  q_i  \\
= (1+  \sum q_i  v_i )+ \sum_i  \big(   \sum_{k\in K_i} {q_i {p_k} v_k}   \big)  +   1 \\
\leq  v +   \sum_i  \big(   \sum_{k\in K_i} {q_i {p_k} v_k}   \big)      +  { 1}\\ 
=w
\end{align*}}
where 	  $\sum_{\iI}  q_i =1$ because $P\red \mdist{q_i P_i}_{\iI}$.

%
%
%

%
%

\item[(iii)]  Suppose  that  some, but not all,  $\Pi_i$ are as  (1.). We \emph{replace} those $\Pi_i$ with  the following other type  derivation of $N_i$
\[\infer{\oder {w_i'=1}  \lett {N_i }M :\zero}{\oder 0  {N_i}:\zero}\]
and use  the argument (ii).
In doing so, we are overapproximating the weight $w_i'=0$ of the original derivation $\Pi_i$ with  ${w_i}'=1$, which is fine, since ultimately we want to obtain a derivation $\Pi$ whose weight is $w\geq 1+ \sum_i q_iw_i $.

\end{itemize}
Case (ii) is exactely the reason why we have $w\geq 1+\sum q_i w_i$ instead of $w= 1+\sum q_i w_i$.
\end{itemize}

\end{proof}

\section{More Variations and  Design Choices}

\subsection{How much to increment the counter?}

As we already observed, in the typing system, the $\letr$-rule increase the counter of $1$ (to reflect a $\letr C$-step) even when such a step never happen. 
As a consequence, in  Weighted Subject Reduction, and therefore in the Completeness Theorem, there is a slight asymmetry:
for  $k\in Nat$, there is a proof $\Pi\dem \oder w M:\at $ such that (1) $ \norm \at = \eval k M $ and (2) $w \geq \etime k M$.

It is indeed possible to give a rule which accurately reflects all  steps, as follows.
\[
\infer[\letr]{\Gamma  \uplus_k  p_k\for\Delta_k \oder {w+ \sum_k p_k w_k+ \RED{ \sum_k p_k}} \lett  N M: {\dsum_k~ p_k \bt_k}}
{\Gamma \oder {w} N: \dist { p_k \mA_k }_{\kK}  	  &	(\Delta_k,  x:  \mA_k\oder {w_k} M: \bt_k)_{k\in K}	   }	
\]	
We have decided to privilege simplicity.

Choosing the  accurate counter $w+ \sum_k p_k w_k+ \RED{ \sum_k p_k}$ would lend, both in the Completeness Theorem and in Weighted Subject  Expansion, 
an equality for both $\at$ and $w$.

\begin{example}[Two ways of counting] This example pinpoints the technical difference between the two design choices for the counter in the  $\letr$-rule, and the consequence in Weighted Subject Reduction. The difference appears when we study the  $\letr C$ case in the proof of Subject Reduction, precisely the case (iii).
	 
Consider the case where $\lett NM\red \mdist{\two \lett {N_1} M, \two \lett {N_2} M}$. $\Pi_1$ and $\Pi_2$ are
	\[ \Pi_1\dem \oder {v_1=0} \lett {N_1}M:\zero    \quad\quad    \Pi_2 \dem  {\oder{v_2=u+1}  \lett {N_2}M:\zero} \]
	
	Therefore \blue{$\two v_1+\two v_2+1= \frac{u}{2}+\two +1$.}
	
	Desired claim: there exists a derivation  $\Pi\dem \oder {w} \lett {N}M:\zero $ where $w= \two v_1+\two v_2+1$
	\begin{enumerate}
		\item {\textbf{Counter +1}}. $\Pi_1$ and $\Pi_2$ are as follows:
		\[ \oder {0} \lett {N_1}M:\zero    \quad\quad    \infer {\oder{u+1}  \lett {N_1}M:\zero}{\oder {u} N_2:\zero} \]
		
		By \ih (using $\oder 0 N_1:\zero$) there is a derivation of $\oder {\frac{u}{2}+1} N:\zero$, from which we obtain 
		\begin{center}
			$ 	\infer {\oder{w=\frac{u}{2}+1+1}  \lett {N}M:\zero}{\oder {\frac{u}{2}+1}  N:\zero}  $
			
		\end{center}
		{	That is, $w>  \two v_1+\two v_2+1$ }
		
		\item {\textbf{Accurate counter}}.
		$\Pi_1$ and $\Pi_2$ are as follows:
		\[ \infer {\oder{0}  \lett {N_1}M:\zero}{\oder {0} N_1:\zero}    \quad\quad    \infer {\oder{u+1}  \lett {N_2}M:\zero}{\oder {u+1} N_2:\zero} \]
		
		By \ih (using $\oder 0 N_1:\zero$) there is a derivation of $\oder {\frac{u}{2}+\two +1} N:\zero$, from which we obtain 
		\[\infer {\oder{w=\frac{u}{2}+\two+1}  \lett {N}M:\zero}{\oder {\frac{u}{2}+\two +1}  N:\zero} \]
		
		That is, $=  \two v_1+\two v_2+1$.
		
	\end{enumerate}
	
\end{example}

\end{document}